\RequirePackage{luatex85}

\documentclass[11pt,letterpaper]{article}
\usepackage{iftex}

\ifPDFTeX
\usepackage[utf8]{inputenc}
\usepackage[T1]{fontenc}
\fi

\usepackage{microtype}

\usepackage{amsmath,amssymb,wasysym,dsfont}
\usepackage{amsthm}
\usepackage{mathtools,float}
\usepackage{xcolor,graphicx,units}
\usepackage{todonotes}
\usepackage{enumitem}
\usepackage{algorithm}
\usepackage[noend]{algpseudocode}
\usepackage{thmtools}
\usepackage{lmodern}
\usepackage{microtype}

\ifPDFTeX

\usepackage{calrsfs}
\DeclareMathAlphabet{\pazocal}{OMS}{zplm}{m}{n}

\usepackage{libertine}
\usepackage[scaled=0.96]{zi4} 
\else
\usepackage{fontspec}
\usepackage{unicode-math}
\fi

\usepackage[paperwidth=8.5in,paperheight=11in,left=1in,top=1in,right=1in,bottom=1in]{geometry}

\declaretheorem[numberwithin=section,refname={Theorem,Theorems},Refname={Theorem,Theorems},name={Theorem}]{theorem}
\declaretheorem[numberlike=theorem,refname={Lemma,Lemmas},Refname={Lemma,Lemmas},name={Lemma}]{lemma}

\declaretheorem[numberlike=theorem,refname={Observation,Observations},Refname={Observation,Observations}]{observation}

\usepackage{enumitem}
\usepackage{hyperref}
\usepackage{graphicx}
\usepackage[utf8]{inputenc} 
\usepackage{amsmath,amssymb,amsthm} 
\usepackage{url}
\usepackage{mathtools}
\usepackage{booktabs}
\usepackage{todonotes}

\makeatletter
\newlength{\within@wd}
\DeclareRobustCommand{\within}[3][c]{%
  \ifmmode%
    \mathchoice%
      {\within@math\displaystyle{#1}{#2}{#3}}%
      {\within@math\textstyle{#1}{#2}{#3}}%
      {\within@math\scriptstyle{#1}{#2}{#3}}%
      {\within@math\scriptscriptstyle{#1}{#2}{#3}}%
  \else%
    \within@do{#1}{#2}{#3}%
  \fi%
}
\newcommand{\within@math}[4]{\within@do{#2}{$\m@th#1#3$}{$\m@th#1#4$}}
\newcommand{\within@do}[3]{%
  \settowidth{\within@wd}{#2}%
  \makebox[\within@wd][#1]{#3}%
}
\makeatother

\DeclareMathOperator{\dist}{d}

\DeclareMathOperator{\isabundant}{\operatorname{A}}
\DeclareMathOperator{\enabledbelow}{\mathrm{enblbelow}}
\DeclareMathOperator{\isopen}{\operatorname{O}}
\DeclareMathOperator{\isenabled}{\operatorname{E}}

\DeclareMathOperator{\switched}{\mathtt{switched}}

\DeclareMathOperator{\parent}{\operatorname{parent}}
\DeclareMathOperator{\children}{\operatorname{children}}

\DeclareMathOperator{\neighborsabove}{\mathtt{neighbors\_above}}

\DeclareMathOperator{\numradii}{{\Delta}}
\DeclareMathOperator{\openbelow}{{\mathtt{openbelow}}}
\DeclareMathOperator{\cost}{{\operatorname{c}}}

\newcommand{\minR}{{\rho_{\mathrm{min}}}}
\newcommand{\maxR}{{\rho_{\mathrm{max}}}}

\newcommand{\Areas}{\mathcal{A}}
\DeclareMathOperator{\X}{X}
\DeclareMathOperator{\Y}{Y}

\newcommand{\Clients}{\mathcal{C}}
\newcommand{\Facilities}{J}

\newcommand{\cround}[1]{{\ceil{\log_5{#1}}}}

\newcommand{\rarea}{{\mathrm{area}}}
\newcommand{\rball}{{\mathrm{x}}}
\newcommand{\rroot}{{\mathrm{root}}}
\newcommand{\rdesg}{{\mathrm{aux}}}
\newcommand{\renbl}{{\mathrm{enabled}}}
\newcommand{\renbll}{{\mathrm{enbl}}}
\newcommand{\rlex}{{\mathrm{lex}}}

\newcommand{\rnew}{{\mathrm{*}}}

\newcommand{\ropen}{{\mathrm{open}}}
\newcommand{\ropt}{{*}}
\newcommand{\OPT}{{\mathrm{OPT}}}
\newcommand{\Solution}{{\mathcal{I}}}
\newcommand{\Pairs}{{\Pi}}
\newcommand{\T}{{\mathcal{T}}}
\newcommand{\fmin}{f_{\mathrm{min}}}
\newcommand{\fmax}{f_{\mathrm{max}}}

\DeclarePairedDelimiterX{\floor}[1]{\lfloor}{\rfloor}{#1}
\DeclarePairedDelimiterX{\ceil}[1]{\lceil}{\rceil}{#1}
\DeclarePairedDelimiterX{\prn}[1]{(}{)}{#1}
\DeclarePairedDelimiterX{\brc}[1]{\{}{\}}{#1}
\DeclarePairedDelimiterX{\brk}[1]{[}{]}{#1}
\DeclarePairedDelimiterX{\abs}[1]{\lvert}{\rvert}{#1}
\DeclarePairedDelimiterX{\tuple}[1]{\langle}{\rangle}{#1}
\DeclarePairedDelimiterX{\set}[2]{\{}{\}}{#1\,\delimsize|\,\mathopen{}#2}

\title{A Tree Structure For Dynamic Facility Location\thanks{A preliminary version of this manuscript appeared in proceedings of the 26th Annual European Symposium on Algorithms (ESA) 2018. The research leading to these results has received funding from the European Research Council under the European Union's Seventh Framework Programme (FP/2007-2013) / ERC Grant Agreement no. 340506.}}

\author{
	Gramoz Goranci\thanks{
		University of Vienna, Austria} 
	\and 
	Monika Henzinger\thanks{
		University of Vienna, Austria}
	\and
	Dariusz Leniowski\thanks{
		University of Vienna, Austria}
}

\begin{document}

\maketitle

\begin{abstract}
We study the metric facility location problem with client insertions and deletions. This setting differs from the classic dynamic facility location problem, where the set of clients remains the same, but the metric space can change over time. We show a deterministic algorithm that maintains a constant factor approximation to the optimal solution in worst-case time $\tilde O(2^{O(\kappa^2)})$ per client insertion or deletion in metric spaces while answering queries about the cost in $O(1)$ time, where $\kappa$ denotes the doubling dimension of the metric.
For metric spaces with bounded doubling dimension, the update time is polylogarithmic in the parameters of the problem.

\end{abstract}

\section{Introduction}

In the \emph{metric facility location problem}, we are given a (possibly infinite) set $V$ of {\em potential clients} or {\em points},  a finite set $\Clients$ of ({\em live clients}),
a finite set $\Facilities\subseteq V$ of {\em facilities} with an \emph{opening cost} 
$f_j$, for each facility $j$, and a metric $\dist$ over $V$, such that $\dist(i,j)$ is the cost of assigning
client $i$ to facility $j$. The goal is to determine a subset $\Facilities' \subseteq \Facilities$ of {\em open} facilities and to assign 
each  client to an open facility such as to minimize the total cost. Obviously it is best to assign each live client to the closest \emph{open} facility.
Thus, the goal can be written as minimizing the objective function
$\sum_{j \in \Facilities'} f_j + \sum_{i \in  \Clients} \min_{j \in \Facilities'} \dist(i,j)$. 

The facility location problem is one of the central problems in combinatorial optimization and operations research~\cite{drezner2001facility}, with many real-word applications. Typical examples include placements of servers in a network, location planning for medical centers, fire stations, restaurants, etc. From the computational perspective, this problem is NP-hard and it is even hard to approximate to a factor better than
$1.463$~\cite{guha1999greedy,sviridenko2002improved}. The best-known polynomial-time algorithm achieves a 1.488-approximation~\cite{li20131}.  

In many applications of facility location, problem data are continuously changing. This has lead to the study of this problem in different settings, e.g., online~\cite{meyerson2001online,anagnostopoulos2004simple,
fotakis2007primal,anthony2007infrastructure,fotakis2008competitive,
nagarajan2013offline,abshoff2015towards}, streaming~\cite{Indyk04,Fotakis11,LammersenS08,CzumajLMS13} or dynamic~\cite{wesolowsky1973dynamic,chardaire1996solving,farahani2009dynamic,
eisenstat2014facility}. The focus of this paper is on the dynamic setting, motivated by mobile network applications, where the set of clients may change over time, and we need to maintain the set of opened servers so as to obtain a solution of small cost after each change. 

Formally, in the \emph{dynamic facility location} problem, the set of clients $\Clients$ evolves over time and queries about both the cost as well as the set of opened facilities can be asked. Specifically,  at each timestep $t$, either a new client is added to $\Clients$, 
a client is removed from $\Clients$, 
a query is made for the approximate cost of an optimal solution 
({\em cost query}), or a query asks for the entire current solution 
({\em solution query}). 
The goal is to maintain a set of open
facilities that after each client update minimizes the above cost function. Thus, the cost $f_j$ of each facility can be seen as a {\em maintenance cost}
that has to be paid for each open facility between two client updates.

\paragraph*{Our contribution.} In this paper we present a \emph{deterministic} data-structure that maintains a $O(1)$-factor approximation algorithm for the metric facility location problem, while  supporting insertions and deletions of clients in $\tilde O(2^{O( \kappa^2)})$ update time, and answering cost queries in $O(1)$ time, where $\kappa$ is the {\em doubling dimension}\footnote{The doubling dimension of a metric space $(V,d)$ is bounded by $\kappa$ if for any $x \in V$ and any radius $r$, any ball 
with center $x$ and radius $r$
in $(V,d)$ can be completely covered by  $2^{\kappa}$ balls of radius $r/2$.} of the metric
space. As the running time per client update is bounded by $\tilde O(2^{O(\kappa^2)})$, the number of
changes in the client-facility assignments is also bounded by this function. For metric spaces with bounded doubling dimension, such as 
the Euclidean space, the running time is $\tilde O(1)$. Formally, we have the following theorem.

\begin{theorem} \label{thm: mainTheorem}
There exists a deterministic algorithm for the dynamic facility location problem when clients and facilities live in a metric space with doubling dimension $\kappa$, such that at every time step the solution has cost at most $O(1)$ times the cost of an optimal solution at that time. The worst-case update time for client insertion or deletion is $O(2^{O(\kappa^2)} \cdot \Delta^3 \cdot (\kappa^2 + \log \Delta))$, where $\Delta$ is logarithmic in the paramters of the problem. A cost query can be answered in constant time and a solution query in time linear in the size of the output.
\end{theorem}

\paragraph*{Comparison with prior work.} The closest work related to our problem is the {\em streaming} algorithm for the {\em metric} facility location problem with \emph{uniform} opening costs due to Lammersen and Sohler~\cite{LammersenS08}. Specifically, given a sequence of insert and deletion operations of points~(clients) from $\{1,\ldots, \Delta\}^{d}$, they devise a Monte-Carlo \emph{randomized} algorithm that processes an insertion or deletion of a point in $\tilde{O}(2^{O(d)})$ time,  using poly-logarithmic space and maintaining a $\tilde{O}(2^{O(d)})$-factor approximation. Since the $d$-dimensional Euclidean space has doubling dimension linear in $d$, we can also interpret the above result in terms of $\kappa$, i.e., the same bounds hold with $d$ replaced by $\kappa$. An easy inspection of the algorithm in~\cite{LammersenS08} shows that the queries can also be answered anytime in the update sequence in $O(1)$ time. Note that this algorithm heavily relies on randomization and the fact that facilities have uniform opening costs. In comparison our algorithm is (1) \emph{deterministic}, (2) achieves a $O(1)$-factor approximation (independent of doubling dimension) (3) generalizes to any metric of bounded doubling dimension, and (4) supports non-uniform opening costs. 
Furthermore, for the Euclidean plane, i.e., $d = 2$, there is a randomized streaming algorithm that achieves a $(1+\varepsilon)$-approximation with poly-logarithmic space~\cite{CzumajLMS13}. However,  it is not clear whether this algorithm supports fast queries.

Regarding the {\em dynamic facility location problem}, multiple variants can be found in the literature~\cite{wesolowsky1973dynamic,chardaire1996solving,farahani2009dynamic,eisenstat2014facility}.
All these variants are different from ours as they assume that the facilities and clients remain the same, and only {\em the 
distance metric} between clients and facilities can change. Additionally, every time a client switches to
a different facility, a switching cost might be incurred and the goal is to minimize the above sum plus all the switching costs. For the version proposed in~\cite{eisenstat2014facility}, there exists a polynomial time constant factor approximation algorithm~\cite{an2015dynamic}.

There is also a large body of prior work on {\em online facility location} (see, e.g.,~\cite{meyerson2001online,anagnostopoulos2004simple,fotakis2007primal,anthony2007infrastructure,fotakis2008competitive,nagarajan2013offline,abshoff2015towards}),
where the clients arrive in an online fashion and have to be connected to a facility, potentially opening new facilities.
If earlier decisions cannot be reversed, then no efficient constant factor approximation algorithm is possible~\cite{fotakis2008competitive}.
If earlier decisions are, however, not permanent, specifically if facilities are only opened for a given amount of time, i.e. {\em leased},
and $k$ different lease lengths are possible~\cite{anthony2007infrastructure}, the offline version of the problem has a polynomial time 3-approximation
algorithm~\cite{nagarajan2013offline}, but the online setting cannot have an approximation better than $\Omega(\log k)$, even for randomized algorithms~\cite{meyerson2001online}.
All of these results are different from ours: (1) We have only one lease length, namely length 1, as each facility can be closed or opened after each timestep,
(2) we allow client arrival {\em and} departure, and (3) our algorithm processes a client update in $O(2^{O(\kappa^2)} \cdot \Delta^3 \cdot (\kappa^2 + \log \Delta))$ worst-case time per operation, where
$\Delta$ is logarithmic in the parameters of the problem, while the running time of the online algorithms is at least linear in the number of facilities~\cite{meyerson2001online,anagnostopoulos2004simple}.

\paragraph*{Technical contribution.} From a technical point of view we modify and significantly extend a hierarchical partition of a  subset of the facilities that was recently introduced for a related problem, called the
{\em dynamic sum-of-radii clustering} problem~\cite{HenzingerLM17}. In that work, a set of facilities $\Facilities$ and a dynamically changing clients $\Clients$ are given and the goal is to
output a set $\Facilities' \subseteq \Facilities$ together with a {\em radius} $R_j$, for each $j \in \Facilities'$, such that the $J'$ covers $\Clients$ and the function
 $\sum_{j \in \Facilities'} (f_j + R_j)$ is minimized. In~\cite{HenzingerLM17} a $O(2^{2\kappa})$-approximation algorithm with time $\tilde O(2^{6\kappa})$ per client
 insertion or deletion is presented for metrics with doubling dimension $\kappa$. Note that the function that is minimized is different from the function minimized in the facility location problem, as the term $\sum_{i \in  \Clients} \min_{j \in \Facilities'} \dist(i,j)$ is replaced by $\sum_{j \in \Facilities'} R_j$. 
 
 More specifically, the hierarchical decomposition of~\cite{HenzingerLM17}  picks a well-separated subset of $J$ with ``small'' cost, assigns one or multiple radii to the selected facilities, and then hierarchically orders the pairs $\tuple{j,R}$ in a tree structure, where $j$ is a facility and $R$ is a radius, such that the children of every pair have a smaller radius and the ``ball'' of the given radius of a child is fully contained in the ``ball'' of its parent with its radius.
 To achieve our result, in Sections~\ref{sec:2} and~\ref{sec:3} we extend this decomposition as follows:
\begin{enumerate}
\item  {\em Abundance condition.} Instead of selecting facilities with ``small''  cost, we introduce the notion of an ``abundance condition'': Facilities that have ``enough nearby'' clients fulfill this condition, and  we {\em only open such facilities}. This leads to following rough notion: The abundance condition is fulfilled for a facility $j$ and a radius $R$ if the number of clients within radius $R$  of $j$ is at least $f_j/R$. The fundamental idea is then as follows: (A) We assign a {\em payment}  of $R$ to each client  within radius $R$ of an open facility $j$, which implies  that the sum of $f_j$ plus the distances of these clients to $j$ is upper bounded by twice the sum of the payments of these clients. 
 (B) We then show that (i) each client pays for at most one facility and (ii) the sum of the
 client payments is linear in the cost of the optimal solution.
 
 \item  {\em Designated facilities.}  To further reduce the cost of the open facilities, we designate to each facility $j$ the ``cheapest nearby'' facility $j^*$, and if a facility $j$ fulfills all conditions to be opened, we open the  facility $j^*$ instead. This allows us  to modify the (rough) abundance condition so that it is  fulfilled for a facility $j$ and a radius $R$ if the number of clients within radius $R$ of $j$ is at least $f_{j^*}/R$.
 This modification increases the distance of the clients only by a constant factor, but might significantly decrease the cost of the open facilities.
 Note that our approach would not achieve a constant factor without this technique: The hierarchical decomposition is based on a well-separated subset of facilities whose construction ignores the facility costs. Thus it might happen that the chosen facilities have high cost,  even though there are ``cheap'' facilities nearby. These cheap nearby facilities are now captured by the designated facilities.
 
 \item {\em  Enabled facilities.} The idea of not opening facilities that are ``close'' to an open facility can be further combined with the hierarchical decomposition.
 More specifically, if facility $j$ with radius $R$ and facility $j'$ with radius $R'$ are ``close'', only one of them is opened,
 namely the lowest-in-the-hierarchy facility that fulfills the abundance condition and has no nearby facility of smaller radius that is
 already open. The advantage of this scheme is that when a facility switches from open to closed or vice versa, we can bound the number of facilities that are
 affected by this change by a function that only depends on $\kappa$. If we had chosen to open the facility with larger radius, 
 then the number of facilities that are affected by
 opening or closing one facility might have been large, i.e., not bounded by a function of $\kappa$ alone.

 \item {\em  Coloring.} 
 Recall that we need to guarantee that each client only pays for {\em one} open facility. 
 To do so, we assign
a {\em color} to each pair $\tuple{j,r}$ in the hierarchy such that no two pairs  $\tuple{j,r}$ and $\tuple{j',r}$, where the distance between $j$ and $j'$ is small,
have the same color. As the metric has bounded doubling dimension, $2^{5 \kappa}+1$ colors suffice for this coloring. 
Then we require that a facility is opened only if it fulfills the abundance condition, {\em and} no facility of either
smaller radius or the same radius but with ``smaller'' color is already open. This requires to further relax the notion of ``closeness'' but reduces the number of open facilities enough so
that each client can be assigned to pay for at most one ``close'' open facility and still every open facility has enough clients paying for its cost.

\end{enumerate}

In Section~\ref{sec:3}, Theorem~\ref{thm:approx}, we show that (a) for clients that are ``close'' to an open facility $j$ in the {\em optimal solution} the sum of their payments (in our solution) is linear in $f_j$
and (b) for clients that are ``far'' from an open facility  in the {\em optimal solution} their payments (in our solution) are within a constant factor of their distance in
the optimal solution. Additionally, we give a data structure that maintains this solution efficiently under insertions and deletion of clients~(see Section~\ref{sec:data_structure}). All missing proofs can be found in appendix.

\section{Preprocessing phase}\label{sec:2}

Let $W$ be the diameter of the metric space, i.e., $\dist(i,j) \le W$, for all $i,j \in V$, let $\fmax = \max_{j \in \Facilities} \{f_j\}$ be the maximum facility opening cost,
and let
$\fmin = \min \set{f_j}{j \in \Facilities, f_j > 0} > 0$ be the minimum opening cost of any facility with non-zero opening cost.
This is w.l.o.g. as all facilities with 0 opening cost are always kept open.
Given the set of clients $\Clients\subseteq V$, let $\OPT = \OPT(\Clients)$ denote the cost of an optimum solution for $\Clients$. In what follows, for the sake of exposition, we also let $\Clients$ to refer to the \emph{current} set of clients.

The algorithm will maintain a number $n = 5^{\lfloor \log_5 |\Clients| \rfloor}$, i.e., the largest power of 5 smaller than  $|\Clients|$: Initally we set $n$ to 0 and use this value of $n$ 
during preprocessing. Whenever the first client is inserted, we set $n = 1$. Afterwards, whenever the number of clients is a factor 5 larger, resp.~smaller, than 
$n$, we update $n$ by multiplying, resp.~dividing it by 5. 

Now let\footnote{Note that $\minR$ could be negative, but it is well-defined as $\fmin > 0$.} $\minR = \cround{(\fmin/{\max{(|\Facilities|}, n)})}$
and let $\maxR = \cround{(\max {(W, \fmax)})} = O(\log W + \log \fmax)$. 
Note that a change of $\minR$ will require an update in our data structures, but this will only happen after $\Theta(n)$ many client insertions or deletions. As we will see later, the cost of this update will be charged against these client updates, and thus does not affect our running times. A \emph{logradius} is an integer $r$ such that $\minR \leq 5^r \leq \maxR$. Let $\numradii = \maxR - \minR + 1$ be the  number of different logradii. Note that
$\numradii = O(\log W + \log (\fmax/\fmin) + \log |\Facilities| + \log |\Clients|)$.  
Finally, let $c_1=20$, $c_2=35$, $c_X = 2c_2 + 2 = 72$, $c_3 = c_X + c_2 = 107$, $c_Y = 2c_3 + c_2 = 249$ and $c_4 = c_Y + c_2 = 284$.

A large part of our data structure is concerned with reducing the number of facilities that are potentially opened and finding an assignment of each client to at most one open facility. 
This is done in multiple ways, as described next.
Based on the approach of~\cite{HenzingerLM17}, we  construct a set of pairs $\Pairs \subseteq (\Facilities \times [\minR, \maxR])$, consisting of facility-logradius pairs and 
a laminar family of \emph{areas}. Different from~\cite{HenzingerLM17}, we color pairs in $\Pairs$,
turning them into triplets, and introduce designated facilities, before defining open, closed and enabled triplets.

\paragraph*{Maximal subsets of distant facilities.}
The first step is to filter out facilities that are close to other facilities. To achieve this, we greedily 
construct a set $\Pairs$ of \emph{pairs} $\tuple{j,r}$ where $j$ is a facility and $r$ is a logradius,  satisfying the following properties:
\begin{enumerate}
\item (Covering) For every facility $j\in \Facilities$ and every logradius $r$,
there exists a facility $j'\in \Facilities_{r}$ with $\dist(j,j')\leq c_1 \cdot 5^{r}$.

\item (Separating) For all distinct $j,j'\in J_r$, $\dist(j,j')>c_1 \cdot 5^r$.
\end{enumerate}
We construct $\Pairs$ as follows: For each logradius $r \in [\minR, \maxR]$, let $\Facilities_r$ be a maximal subset of $\Facilities$ 
such that any two facilities in $\Facilities_r$ are at distance strictly larger than $c_1\cdot 5^r$. Set $\Pairs \gets \bigcup_r \set{\tuple{j,r}}{j \in \Facilities_r}$. Note that for $r=\maxR$, the set $\Facilities_r$ contains just one facility.

\paragraph*{Hierarchical decomposition of $\Pairs$.}
We now construct a hierarchical decomposition of $\Pairs$ and represent it by a tree $\mathcal{T}$, using the following algorithm. Set the root of $\mathcal{T}$ to be the unique pair $\tuple{j,\maxR}$. For each $r<\maxR$ and $j\in \Facilities_r$: (1) Set $j'\in \Facilities_{r+1}$ be the facility closest to $j$. (2) Set $\parent(j,r)\gets \tuple{j',r+1}$.

By construction, $\mathcal{T}$ has height at most $\Delta$ and the parent of a pair $\tuple{j,r}$ is a pair of the form $\tuple{j',r+1}$. The following three lemmas describe the crucial properties of the tree $\mathcal{T}$.  
\begin{lemma}[Nesting of balls]\label{lemma:inclusion-of-big-balls_2}
 Let $c^*$ be any constant such that $c^* \geq (5/4)c_1$. If $\parent(j,r)=\tuple{j',r+1}$, then $\dist(j,j') \le c_1 \cdot 5^{r+1}$ and $B(j,c^*\cdot 5^r)\subseteq B(j',c^*\cdot 5^{r+1})$.
\end{lemma}

\begin{lemma}[\cite{HenzingerLM17}]\label{lem:doubling_dimension_2}
For any point $p$, radius $r$ and some number $\alpha > 0$, the set of pairs
$\Pairs(p,r) = \set{\tuple{j,r} \in \Pairs}{\dist(p,j) < 2^{\alpha}c_1\cdot 5^r}$
has at most $2^{(\alpha +1)\kappa}$ elements, where $\kappa$ is the doubling
dimension of the metric space.
\end{lemma}

\begin{lemma}[\cite{HenzingerLM17}]\label{lemma:degree_2}
A node $\tuple{j,r}$ of $\mathcal{T}$ has at most $2^{4\kappa}$ children
\end{lemma}

\paragraph*{Hierarchical decomposition of $V$ into a laminar family of areas.}
The balls $B(j,r)$ and $B(j',r)$ with $\tuple{j,r}$ and $\tuple{j',r}$ in $\Pairs$ might overlap, which is problematic for the mapping of clients to facilities. To rectify this problem, we partition $V$ into a laminar family of \emph{areas} such that no two same-logradius areas overlap, as follows: For each $\tuple{j,r}\in\Pairs$, initialize $A(j,r)\gets\emptyset$. Next, for each point $p\in V$: (1) Let $r^*$ be the smallest such that there exists pairs $\tuple{j,r^*} \in \Pairs$ with $p\in B(j,c_2\cdot 5^{r^*})$. (2) Among all such pairs, let $\tuple{j^*,r^*}$ denote the one minimizing $\dist (p,j^*)$. (3) Add $p$ to the set $A(j^*,r^*)$ and to every set $A(j',r')$ with $(j',r')$ ancestor of $(j^*,r^*)$ in $\mathcal{T}$. 

The laminar family of areas fulfills the following lemmas.
\begin{lemma}[\cite{HenzingerLM17}] \label{lemma:area-contained-in-big-ball_2}
For each $\tuple{j,r} \in \Pairs$, $j \in A(j,r)\subseteq B(j,c_2\cdot 5^r)$ and if $\parent(j,r)=\tuple{j',r+1}$, then $A(j,r) \subseteq A(j', r+1)$.
\end{lemma}

\begin{lemma}\label{lemma:covering-balls-by-areas_2} Let $j \in J$, 
$r \in [\minR, \maxR]$ and $p \in B(j,5^{r})$. Then there exists a pair $\tuple{j',r} \in \Pairs$ such that $p \in  A(j',r)$ and $\dist(j,j') \leq (c_2 + 1) \cdot 5^{r}$.
\end{lemma}
Additionally an even stronger statement regarding the points covered by areas versus points covered by balls holds:
\begin{lemma}\label{lem:areas-equal-balls_2}
For each logradius $r \in [\minR,\maxR]$, 
$\bigcup_{\tuple{j,r} \in \Facilities_r}A(j,r)=\bigcup_{\tuple{j,r} \in \Facilities_r}B(j,c_2\cdot 5^r).$
\end{lemma}

\paragraph*{Covering balls using unions of areas.} 

To select which facilities with a pair $\tuple{j,r}$ in $\Pairs$ to open, we introduce below the abundance condition which measures how many clients are ``close'' to $j$. For measuring ``closeness'', we would like
to say that a client $i$ is close to a facility $j$ if $i \in \X(j,r)$ for some definition of $\X(j,r)$ that fulfills the crucial property that for every $\tuple{j,r}$ there exists a pair $\tuple{j',r} \in \Pairs$ such that $B(j, 5^r) \subseteq \X(j',r)$.
Note that this might not hold if we use $\X(j,r) = A(j,r)$ and it
 does not follow from Lemma~\ref{lemma:covering-balls-by-areas_2} as different points of $B(j, 5^r)$ might belong to different areas $A(j',r)$, whose facilities might be up
to distance $(2 c_2 + 2) \cdot 5^r$ apart. Thus, for any $\tuple{j,r} \in \Pairs$, we define $\X(j,r)$ as follows: 
\[ \X(j,r) = \bigcup\set{A(j',r)}{\dist(j,j') \leq c_X \cdot 5^{r}}, \quad\text{ where } c_X = 2c_2 + 2, \] 
and can now show the desired property for $\X(j,r)$:
\begin{lemma} \label{lem: cover-with-x}
Let $j \in J$ with $\tuple{j,r} \not \in \Pairs$. Then there exists $\tuple{j^{*},r} \in \Pairs$ with $B(j,5^{r}) \subseteq \X(j^*,r)$.
\end{lemma}

\begin{proof} It suffices to show that there exists $\tuple{j^{*},r} \in \Pairs$ such that for every area $A(j',r)$ that intersects with the ball $B(j,5^r)$, we get that $A(j',r) \subseteq \X(j,r)$. First, by the Covering property, there exists a pair $\tuple{j^{*},r} \in \Pairs$ such that $\dist(j,j^*) \leq c_1 \cdot 5^{r}$. Next, let $A(j',r)$ be any area such that $A(j',r) \cap B(j,5^{r}) \neq \emptyset$. By Lemma~\ref{lemma:covering-balls-by-areas_2}, we get that $\dist(j',j) \leq (c_2+1)\cdot 5^{r}$. It follows that $\dist(j',j^*) \leq \dist(j',j) + \dist(j,j^*) \leq (c_1 + c_2 + 1) \cdot 5^{r} \leq c_X\cdot5^{r}$, which in tun implies that $A(j',r) \subseteq \X(j^*,r)$.
\end{proof}
It is crucial for the running time to get a bound on the number of areas used to construct $\X(j,r)$. We do this in the following lemma, which is a simple corollary of Lemma~\ref{lem:doubling_dimension_2}.

\begin{lemma} \label{lem: x-nr-areas}
For any $\tuple{j,r} \in \Pairs$, $\X(j,r)$ is a union of at most $2^{3\kappa}$ areas.
\end{lemma}
We also need to bound how far any two points in $\X(j,r)$ can be apart:

\begin{lemma} \label{lem: x-cont-ball}
It holds that $\X(j,r) \subseteq B(j,c_3\cdot 5^{r})$.
\end{lemma}

To further reduce the set of open facilities, it is necessary to introduce a notion of  ``closeness'' between facilities that is more relaxed than the definition
used for covering, where we required two facilities to be at least $c_1 \cdot 5^r$ apart. Now we guarantee that if a facility is open then no
other facility within distance $c_Y \cdot 5^r$ is opened, resulting in the following definition of $\Y(j,r)$ for every $\tuple{j,r} \in \Pairs$:

\[ \Y(j,r) = \bigcup\set{A(j',r)}{\dist(j,j') \leq c_Y\cdot 5^{r}}, \quad \text{ where } c_Y = 2c_X + 3
c_2.\]
The constant $c_Y$ is chosen so that we can make sure that there are never two pairs $\tuple{j,r}$ and $\tuple{j', r'}$ such that both $j$ and $j'$ are open
and $\X(j,r)$ and $\X(j', r')$ intersect. Thus, a client can always only belong to at most one set $\X(j,r),$ where $\tuple{j,r}$ is open. The facility $j$
will be the facility that the client is assigned to. The following lemmas follow as before:

\begin{lemma} \label{lem: y-nr-areas}
For any $\tuple{j,r} \in \Pairs$, $\Y(j,r)$ is a union of at most $2^{5\kappa}$ areas.
\end{lemma}
\begin{lemma} \label{lem: y-cont-ball}
It holds that $\Y(j,r) \subseteq B(j,c_4\cdot 5^{r})$.
\end{lemma}
\begin{lemma} \label{lem: y-of-ancestor}
If $\tuple{j',r'}$ is an ancestor of $\tuple{j,r}$ in $\T$, then $\Y(j,r) \subseteq \Y(j',r')$.
\end{lemma}

\paragraph*{Coloring of pairs.}
We are now ready to define a tie-breaking rule based on colors.
For any $\tuple{j,r} \in \Pairs$, consider the set $\Pairs(j,r) = J_r \cap B(j,c_4 \cdot 5^r)$. By Lemma~\ref{lem:doubling_dimension_2} and since $c_4 < 16 c_1$, it follows that $\Pairs(j,r)$ contains at most $2^{5\kappa}$ pairs.
We need to guarantee that at most one of them will ever be opened. Thus we introduce a tie-breaking rule based on colors of pairs. This guarantees that out of all pairs in 
$\Pairs(j,r)$ that fulfill the abundance condition only the one with the ``smallest'' color is opened.

More formally, we perform a preprocessing step using a greedy approach to color pairs of same log-radius in $\Pairs$ with $2^{5\kappa}+1$ colors from $0$ to $2^{5\kappa}$. 

Specifically for each log-radius $r \in  [\minR, \maxR]$, we greedily color every pair 
$\tuple{j,r}$ of $J_r$ by one color $s$ 
so that 
\emph{no two pairs}
$\tuple{j,r}, \tuple{j', r} \in J_r$ with
$\dist(j, j') \le c_4 \cdot  5^r$ 
\emph{are colored with the same color},
and we refer to $\tuple{j,r,s}$ as \emph{triplet}.  
Let $J_{r,s} := \set{\tuple{j,r,s'} \in J_r}{ s' = s}$.

\paragraph*{Designated facilities.}

Furthermore, even if a triplet $\tuple{j,r,s}$ fulfills the condition to be opened (which is explained in the next section) it will not be opened, if there is a ``cheaper'' facility nearby. More formally, we precompute for each pair $\tuple{j,r}$ in $\Pairs$
the following {\em designated facility}:
Let  $f^*_{\tuple{j,r}}$ be the minimum opening cost of any facility in $\X(j,r)$, i.e.,
$$f^*_{\tuple{j,r}} = \min\set[\big]{f_{j'}}{
    {j'} \in \Facilities \cap \X(j,r)}.$$
 The {\em designated facility} $j^*_{\tuple{j,r}}$ of
$\tuple{j,r}$  is the facility with minimum cost $f^*_{\tuple{j,r}}$ in $\X(j,r)$ with ties broken according to the minimum id-number, i.e.,
$j^*_{\tuple{j,r}} = \min\set[\big]{j'}{j' \in \Facilities \cap \X(j,r), 
f_{j'} = f^*_{\tuple{j,r}}}.$
		
\begin{observation}\label{obs:openf_2}
For any $\tuple{j,r} \in \Pairs$, $\dist(j, j^*_{\tuple{j,r}} ) \leq c_3 \cdot 5^r$ and $f^*_{\tuple{j,r}} \le f_j$.
\end{observation}
If the triplet $\tuple{j,r,s}$ fulfills the condition to be opened, we open $j^*_{\tuple{j,r}}$ instead, or do nothing if $j^*_{\tuple{j,r}}$ is already open.
Whenever all triplets for which a
facility is designated are closed, then the facility is {\em closed}.

\section{Processing updates}\label{sec:3}

After the preprocessing phase we are given a laminar family of triplets, where each triplet $\tuple{j,r,s}$ is formed by an area $A(j,r)$ along with its pre-defined color $s$. Depending on the set of clients $\Clients$,
a triplet can be either \emph{disabled} or \emph{enabled} and either \emph{open} or \emph{closed},
where each open triplet is also enabled.
These properties of triplets are maintained dynamically as the set $\Clients$ of clients changes. 
Initially $\Clients = \emptyset$ and all triplets are closed and disabled.
We now proceed to the formal definitions.

\paragraph*{Open triplets. }
We {\em open} a triplet $\tuple{j,r,s}$ if there are enough clients in the set $X(j,r)$ to pay 
the opening cost and it has no strictly smaller-radius or no same-radius and strictly smaller-color open triplet in its ``neighborhood''.
A triplet that is not open is {\em closed}. Formally a triplet $\tuple{j,r,s}$ is open if it belongs to the set $ \Facilities_{r,s}^\ropen(\Clients)$, 
which is defined\footnote{Remark that to make the formula a bit simpler we slightly abuse notation here -- 
$\Y(j,r)$ is a set of areas (i.e., subsets of the metric space), while $\Facilities_{r',s'}^\ropen(\Clients)$ is a set of triples.
Formally the intersection should be understood as
$\Y(j,r) \cap \set{j' \in \Facilities}{\tuple{j',r',s'} \in \Facilities_{r',s'}^\ropen(\Clients)}$.}
recursively as follows.
\begin{align*}
\Facilities_{r,s}^\ropen(\Clients) = \Big\{ & \tuple{j,r,s} \in \Facilities_{r,s} \ \Big|\ 
    5^r \cdot \abs{\Clients \cap \X(j,r)} \geq 
        f^{\ropt}_{\tuple{j,r}}
    \land \forall\tuple{r',s'} <_\rlex \tuple{r,s} : \\ 
    &   \Y(j,r) \cap \Facilities_{r',s'}^\ropen(\Clients) = \varnothing
      \Big\}.
\end{align*}
We use $\Facilities^\ropen_\Clients = \bigcup_{r,s} \Facilities_{r,s}^\ropen(\Clients)$ to denote the set of all open clients
and  $\Solution_\Clients$ to denote the set of all open facilities, i.e.,
$ \Solution_\Clients = \set{j^*_{\tuple{j,r}} \in \Facilities}{
    \tuple{j,r,s} \in J^\ropen_\Clients}. $
    
We call the following condition for $\tuple{j,r,s}$, used in the definition of $J_{r,s}^\ropen$,
 the \emph{abundance condition},
\begin{equation}
5^r \cdot \abs{\Clients \cap \X(j,r)} \geq f^*_{\tuple{j,r}}
\label{eq:abundance_condition_2}.
\end{equation}
\begin{lemma} \label{lem:oneOpen2}
If $\abs{\Clients} > 0$ then there exists at least one open triplet.
\end{lemma}
When showing the bound on the approximation ratio, we need the
property that for each point $i \in V$ there is {\em at most
one} set $\X(j,r)$ associated with an open triplet such that $i$ belongs to. This is necessary to make sure that each client ``pays'' for at most one open facility.
\begin{lemma}\label{lem:at_most_one_ball}
Each client $i \in \Clients$ belongs to at most one $\X(j,r)$ with $\tuple{j,r, s} \in J^\ropen_{\Clients}$ for some color $s$.
\end{lemma}
Note, however, that is not true that each client $i \in \Clients$ is contained in {\em at least one} $\X(j,r)$ associated with an open triplet for some $r$: even though there always exists a $\X(j,r)$ fulfilling the abundance condition and 
containing $i$ (namely $\X(j^{\rroot}, \maxR)$), the corresponding triplet $\tuple{j^{\rroot},\maxR,s}$ might not be open due to a ``nearby'' open triplet of smaller logradius. To deal with this issue we introduce enabled triplets and show that each client in $\Clients$ is contained in at least one $\X(j,r)$ of an enabled triplet.

\paragraph*{Enabled triplets. } A triplet $\tuple{j,r,s}$ is {\em enabled} if it belongs to the set $ \Facilities_{r,s}^\renbl(\Clients)$, which is defined\footnote{Similarly to the definition of $\Facilities_{r,s}^\ropen(\Clients)$ we mean here
$\Y(j,r) \cap \set{j' \in \Facilities}{\tuple{j',r',s'} \in \Facilities_{r',s'}^\ropen(\Clients)}$.} as follows:
\[
  \Facilities_{r,s}^\renbl(\Clients) = \set[\Big]{\tuple{j,r,s} \in J_{r,s}}{ 
    \exists \tuple{r',s'} \leq_\rlex \tuple{r,s}:
    \Y(j,r) \cap \Facilities_{r',s'}^\ropen(\Clients) \neq \varnothing }. 
\]
We use $ \Facilities^\renbl_\Clients = \bigcup_{r,s} \Facilities_{r,s}^\renbl(\Clients)$ to denote the set of all enabled facilities. The following observation follows from the definition.
 \begin{observation} \label{obs:openEnabled_2}
If a triplet is open, then it is also enabled.
 \end{observation}
Furthermore, as a corollary of Lemma~\ref{lem: y-of-ancestor} we have the following lemma.
 \begin{lemma} \label{lem:ancesEnabled_2}
 If a triplet is enabled, then all its ancestors in $\mathcal{T}$ are also enabled.
\end{lemma}
Since $A(j^{\rroot}, \maxR) = V$, every point in $V$ belongs to at least one $\X(j,r)$ associated with an enabled triplet. To guarantee that at least one triplet is enabled, recall that Lemma~\ref{lem:oneOpen2} showed that if $\abs{\Clients} > 0$, then there exist an open, and, thus, also enabled triplet
(see Observation~\ref{obs:openEnabled_2}). 
Lemma~\ref{lem:ancesEnabled_2} implies that the root of $\T$ is enabled, implying the following lemma.
 \begin{lemma}\label{lem:allenabled}
 If $\abs{\Clients} > 0$ then every point in $V$ belongs to at least one $\X(j,r)$ associated with an enabled triplet.
 \end{lemma} 
The next lemma shows that the definition of enabled implies that any triplet that satisfies the abundance definition is either open or enabled. This is a crucial observation for the proof of the approximation ratio, as it will allow us to argue that for any facility that the optimal solution opens, an enabled triplet must be nearby.

\begin{lemma}\label{obs:abundance_implies_enabled_2}
If $\tuple{j,r,s}$ satisfies the abundance condition, then it is enabled.
\end{lemma}

\paragraph*{Assignment of clients. }
We next describe how to assign each client $i$ to an enabled triplet $\tuple{j,r,s}$ and an open facility. If $\tuple{j,r,s}$ is not open, we show how to find a close open triplet of smallest radius.
For this open triplet we know its designed facility that is open. This is the facility that the client is finally assigned to.

We start with the assignment of $i \in C$ to an enabled triplet. To this end, let $r_i^{\rarea}$ be the minimum logradius of any enabled triplet such that $i$ belongs to the area associated with the triplet, i.e.,
$r^{\rarea}_i = \min\set{r}{
  \tuple{j,r,s} \in \Facilities^\renbl_\Clients, i \in A(j,r)}$. Note that $A(j,r) \subseteq \X(j,r)$, and let $j_i^{\rarea}$ be the center of the area with log-radius $r^{\rarea}_i$  and define $\tuple{j_i^{\rarea}, r^{\rarea}_i,s^{\rarea}_i}$ to be the corresponding triplet (recall that same-logradius areas are disjoint).
  
Once we determined the enabled triplet $\tuple{j_i^{\rarea}, r^{\rarea}_i,s^{\rarea}_i}$ of $i$,
we assign $i$ to the open triplet of minimum radius such that the corresponding facility belongs to $\Y(j_i^\rarea, r_i^\rarea)$ 
and let $j^\ropen_i$ be the designated open facility of that open triplet. We assign $i$ to it. Formally:
\begin{align*}
\tuple{r^\rdesg_i, s^\rdesg_i,j^\rdesg_i} &= \min\Big\{
    \tuple{r',s',j'} \ \Big| \
    \tuple{j',r',s'} \in \Facilities^\ropen_\Clients,
    \tuple{r',s'} \leq_\rlex \tuple{r^\rarea_i,s_i^\rarea}, \\
     & \quad\quad\quad\quad j' \in \Y(j^\rarea_i, r^\rarea_i)\Big\},\\
j^\ropen_i &= j^*_{\tuple{j^\rdesg_i,r^\rdesg_i}}.  
\end{align*}
Finally we denote the set of all clients assigned to facility $j$ by
$\Clients_j = \set{i \in \Clients}{j = j_i^{\ropen}}$. Note that $j_i^\ropen$ does not have to be the closest open facility, but as the next lemma shows it is not far away from an open facility.
\begin{observation}\label{obs:again_nearby_open_facility}
Any $i \in \Clients$ is within $(c_2 + c_3 + c_4) \cdot 5^{r^\rarea_i}$ 
of an open facility.
\end{observation}
The value $r^{\rarea}_i$ is crucial for the cost estimate. Thus it is important to characterize this
value even further, as we do in the following lemma.
\begin{lemma} \label{lem:radiusMatch2}
For any $i \in \Clients$, if $i \in \X(j,r)$ and $\tuple{j,r,s} \in J^\ropen_\Clients$ for some color $s$, then $r^{\rarea}_i = r$. 
\end{lemma}
Assume we open each facility in $\Solution_{\Clients}$ and each client is assigned to an open facility as described above or an even closer one, if one exists.
 As a consequence of the above lemma and the definition of open facilities we can now bound the total cost of the solution by $O(\sum_{i \in \Clients} 5^{r^\rarea_i})$.
\begin{lemma}\label{lm:facilities_costs_still_dont_matter} It holds that $
\sum_{i \in \Clients} \dist(i, j_i) + \sum_{j \in \Solution_{\Clients}} f_j
  \leq \sum_{i \in \Clients} (c_2 + c_3 + c_4 + 1) \cdot 5^{r^\rarea_i}$. 
\end{lemma}
Now we are ready to prove the bound on the approximation ratio.
\begin{theorem} \label{thm:approx}
For any subset $\Clients \subseteq V$ of clients, assign each client $i \in \Clients$ to the facility $j_i^{\ropen}$. Then the cost of this solution is $O(1) \cdot \OPT$, where $\OPT$ is the optimal solution for $\Clients$.
\end{theorem}
\begin{proof}
Denote by $\Solution^\ropt$ an arbitrary optimal solution. For each $i \in \Clients$ define $
j^\ropt_i $ to be the facility $i$ is connected to. Moreover, for each $j \in \Facilities$ let  $\Clients^\ropt_j$ be the set of clients connected to $j$. Formally,
$
j^\ropt_i = \min\set{j \in \Solution^\ropt}{
\dist(i,j) = \dist(i,\Solution^\ropt)},~
\Clients^\ropt_j = \set{i \in \Clients}{j = j^\ropt_i}.
$
\noindent Consider some $j \in \Solution^\ropt$ and let $r \in [\minR, \maxR]$ be the logradius such that  

\begin{equation}\label{eq:enough_clients_for_j_2}
5^{r-1}\cdot \abs{\Clients^\ropt_j \cap B(j, 5^{r-1})} <
  f_j 
  \leq 5^{r}\cdot \abs{\Clients^\ropt_j \cap B(j, 5^{r})}.
\end{equation}
Note that $r$ is well-defined as $5^{\maxR} \cdot \abs{\Clients^\ropt_j \cap B(j, 5^\maxR)} \ge \fmax \geq f_j$ by the definition of $\maxR$
and $f_j \ge \fmin \ge 5^{\minR + \log_5 n - 1} \ge 5^{\minR-1}\cdot \abs{\Clients^\ropt_j \cap B(j, 5^{\minR-1})}$ by the definition of $\minR$. 

To complete the proof we split $\Clients^\ropt_j$  into two sets,
$\Clients^{\mathrm{lo}}_j$ and $\Clients^{\mathrm{hi}}_j$, according to whether $i \in B(j,5^{r-1})$
or not, i.e., $\dist(i,j) \leq 5^{r-1}$ or $\dist(i,j) > 5^{r-1}$, respectively, and show that
\begin{equation} \label{eq:a} \tag{A}
\sum_{\mathclap{i \in \Clients^{\mathrm{lo}}_j}} 
       5^{r^\rarea_i} 
  <5 \cdot f_j, \text{ and}
\end{equation}
\begin{equation} \label{eq:b}  \tag{B}
\sum_{i \in \Clients^{\mathrm{hi}}_j} 5^{r^\rarea_i} \leq
   \sum_{i \in \Clients^{\mathrm{hi}}_j} 5 \cdot \dist(i,j).
\end{equation}
Before showing the above inequalities we first argue that they prove the bound on the approximation ratio. Note that every client belongs to $\Clients^{\ropt}_j$ for some $j \in \Solution^{\ropt}$. Thus, 
\[ \sum_{i \in \Clients} 5^{r^\rarea_i} \leq \sum_{j \in  \Solution^\ropt} 5  \cdot f_j + \sum_{i \in \Clients} 5 \cdot \dist(i,j) \le 5 \cdot \OPT. \]

Using Lemma~\ref{lm:facilities_costs_still_dont_matter} we get that
\[
\sum_{i \in \Clients} \dist(i,j_i^{\ropen}) + \sum_{j \in \Solution_\Clients} f_{j}
  \leq \sum_{i \in \Clients} (c_2 + c_3 + c_4 + 1) \cdot 5^{r^\rarea_i} 
  \leq 5(c_2 + c_3 + c_4 + 1) \cdot \OPT.
\]
We first show~\eqref{eq:a}. Consider $i \in \Clients^{\mathrm{lo}}_j$, or equivalently, $i \in B(j,5^{r-1})$. We claim that $r^{\rarea}_i \leq r$. Indeed, if $\tuple{j,r,s} \in \Pairs$ for some color $s$, then $\X(j,r) \supseteq A(j,r) \supseteq B(j,5^{r})$ along with \eqref{eq:enough_clients_for_j_2} give
\[ 5^{r} \cdot \abs{\Clients \cap \X(j,r)} \geq 5^{r} \cdot \abs{\Clients_j^{\ropt} \cap B(j,5^{r})} \geq f_j \geq f^{\ropt}_{\tuple{j,r}},
\]
which in turn implies that $\tuple{j,r,s}$ satisfies the abundance condition and so it is enabled (Observation~\ref{obs:abundance_implies_enabled_2}). By definition of $r_i^{\rarea}$, we get $r_i^{\rarea} \leq r$. 
If $\tuple{j,r,s} \not \in \Pairs$ for any $s$, then by Lemma~\ref{lemma:covering-balls-by-areas_2}, there exists a triplet $\tuple{j',r,s}$ such that $i \in A(j',r)$.
Because $\dist(j',j) \leq \dist(j',i) + \dist(i,j) \leq (c_2+1)\cdot 5^{r}$ we have that $B(j, 5^r) \subseteq \X(j',r)$.
This along with~\eqref{eq:enough_clients_for_j_2} imply that
\[ 
5^{r} \cdot \abs{\Clients \cap \X(j',r)} \geq 5^{r} \cdot \abs{\Clients_j^{\ropt} \cap B(j,5^{r})} \geq f_j \geq f^{\ropt}_{\tuple{j',r}},
\]
where the last inequality follows by definition of $f^{\ropt}_{\tuple{j',r}}$. Similarly it follows that $\tuple{j',r,s}$ is enabled and $r_i^{\rarea} \leq r$.
Recalling that $i \in B(j,5^{r-1})$ we finally arrive at
\[
\sum_{\mathclap{i \in \Clients^{\mathrm{lo}}_j}} 
       5^{r^\rarea_i} 
  \leq \sum_{\mathclap{i \in \Clients^{\mathrm{lo}}_j}} 
       5^{r} 
  \leq 5 \cdot 5^{r-1} \cdot \abs{\Clients^\ropt_j \cap B(j, 5^{r-1})} 
  \leq 5 \cdot f_j.
\]
We next show~\eqref{eq:b}. First, observe that for any ball $B(j,5^{r''})$, $r'' \geq r$ with $i \in B(j, 5^{r''})$, one can prove similarly to the above that there exists an \emph{enabled} triplet $\tuple{j'',r'',s}$ such that $i \in A(j'', r'')$. We have that $\dist(j'',j) \leq \dist(j'',i) + \dist(i,j) \leq (c_2+1)\cdot 5^{r''}$, thus $X(j'',r'') \supseteq B(j,5^{r''}) \supseteq B(j,5^r)$. This, together with~\eqref{eq:enough_clients_for_j_2} and $f^{\ropt}_{\tuple{j'',r''}} \leq f_j$ give the claim.

Now, consider $i \in \Clients^{\mathrm{hi}}_j$. Since $\dist(i,j) > 5^{r-1}$, we get that $\cround{\dist(i,j)} \geq r$ and $i \in B(j, 5^{\cround{\dist(i,j)}})$. By the discussion above, there exists an enabled triplet $\tuple{j'',\cround{\dist(i,j)},s}$  such that $\X(j'',\cround{\dist(i,j)}) \supseteq B(j, 5^{\cround{\dist(i,j)}})$, implying that $r_i^{\rarea} \leq \cround{\dist(i,j)}$ and so

\[
\sum_{i \in \Clients^{\mathrm{hi}}_j}5^{r^\rarea_i}
  \leq \sum_{i \in \Clients^{\mathrm{hi}}_j}5 \cdot \dist(i,j).\qedhere
\] 

\end{proof}

\section{Data Structure}
\label{sec:data_structure}
In this section we devise a data structure for the dynamic metric facility location problem that supports insertions and deletions of clients as well as returning (a) the approximate cost of the optimal solution or (b) a set of open facilities that achieves this approximate cost. We achieve this by maintaining the minimum cost solution restricted to pairs in $\Pairs$. By Theorem~\ref{thm:approx} this is a $O(1)$ approximation to the cost of the optimal solution.

From the preprocessing phase the algorithm is given the set $\Pairs$ of facility-radius-color triplets, as well as the laminar family of areas $\Areas$ with its dependency tree $\mathcal{T}$ using the following representation.
(1)  A two-dimensional array of size $(\maxR - \minR + 1)  \times (2^{5\kappa}+1)$, keeping for each logradius $r \in [\minR, \maxR]$ and color $s$ a list of all the facilities of $J_r$ that share the color $s$, and
(2) the dependency tree $\T$ in a tree data structure. Whenever we use the term \emph{subtree}, \emph{child}, or \emph{descendant} in the following we refer to the dependency tree.
(3) For each triplet $v = \tuple{j,r,s} \in \Pairs$, the list $\neighborsabove(v)$ of all triplets $\tuple{j',r',s'}$ such that (a) $\tuple{r',s'} >_\rlex \tuple{r,s}$ and (b) $j \in \Y(j',r')$. 
(4) For each triplet $v = \tuple{j,r,s} \in \Pairs$, the value $f^*_{\tuple{j,r}}$, which is the minimum opening cost among all facilities in $\X(j,r)$. 
Using the algorithm in Subsection~\ref{sec:lookup_ball} each list $\neighborsabove$  and each value $f^*_{\tuple{j,r}}$ can be computed in time $O(2^{O(\kappa)} \Delta)$. Thus, the above data structure
can be built in time $O(\abs{\Facilities} \cdot 2^{O(\kappa)} \Delta)$.

The algorithm will maintain a dynamic data structure, which can be viewed as an \emph{annotated dependency tree} that keeps for each node $v = \tuple{j,r,s}$ of $\mathcal{T}$ the following information:

\begin{enumerate}
\setlength\itemsep{-0.05cm}
\item three bits $\isopen_{\rball}(v)$, $\isenabled_{\rball}(v)$, $\isabundant_{\rball}(v)$, which indicate whether the triplet $\tuple{j,r,s}$ is open, enabled and fulfils the abundance condition, respectively,
\item the number $n_{\rarea}(v)$ of current clients that belong to the area $A(j,r)$, i.e., $n_{\rarea}(v) = \abs{\Clients \cap A(j,r)}$,
\item the number $n_\rball(v)$ of current clients that belong to $\X(j,r)$, i.e., $n_\rball(v) = \abs{\Clients \cap \X(j,r)}$,
\item the number $\openbelow(v)$ of all open triplets $\tuple{j',r',s'}$ with $\tuple{r',s'} <_{\rlex} \tuple{r,s}$ and their corresponding facilities falling within $\Y(j,r)$, i.e., 
  \[ \openbelow(v) = \abs[\Big]{\set[\big]{\tuple{j',r',s'} \in \Facilities^\ropen_\Clients}{ \tuple{r',s'} <_\rlex \tuple{r,s}, j' \in \Y(j,r)}}, \] 
\item the number $n_{\enabledbelow}(v)$ of current clients that belong to areas below that are enabled, i.e., \[ n_{\enabledbelow}(v) = \abs[\Big]{\Clients \cap \bigcup \set[\big]{\tuple{j',r',s'} \in J^\renbll_{\Clients}}{A(j',r')  \subset A(j,r)}},\]
\item the value $\cost(v) = \sum \set{5^{r^\rarea_i}}{i \in \Clients \cap A(j,r)}$ (note that with the currently open facilities the cost accrued for the clients in $A(j,r)$ is $O(\cost(v))$),
\item the value $y(v)$, which is the cost of the children of $v$, i.e., $y(v) = \sum_{u \text{ child of } v} \cost(u)$.

\end{enumerate}
We next describe the usefulness of the information we keep. Points 1-2 are self-explanatory. Point 3 provides information to test the abundance condition, and thus update the bits in Point 1. Point 4 is useful when deciding whether we should open an area or not. Points 5-7 allow us to efficiently update the cost of the solution.

\subsection{Finding all balls containing a given point}\label{sec:lookup_ball}

In this section we describe a crucial subroutine that we use repeatedly when handling updates. It is given the hierarchy data structure for $\mathcal{T}$, an arbitrary point $p \in V$ and some constant $c^*$ such that $c^* \geq (5/4) c_1$ and returns all balls $\tuple{j,r} \in \Pairs$ that are at distance at most $c^* \cdot 5^{r}$ from $p$, i.e., $p \in B(j, c^* \cdot 5^{r})$. For $r \in [\maxR, \minR]$, let $S(r)$ denote the set of such balls. 

The algorithm $\textsc{FindBalls}(p,c^*)$ performs a top-down traversal of the tree starting at its root $\tuple{j, \maxR}$. Note that by the definition of $\maxR$, all points belong to $B(j,c^* \cdot 5^{\maxR})$ and $S(\maxR) = \{\tuple{j, \maxR}\}$.  For computing $S(r)$, where $r = (\maxR-1)$, it determines all children of the root to find the pairs $\tuple{j',r}$ such that the distance of $j'$ and $p$ is at most $c^* \cdot 5^{r}$. This step is repeated to compute the set $S(\ell)$ for every level of the hierarchy, until we reach the bottom-most level. Finally, we let $S \coloneqq \bigcup \{S(r) : r \in [\minR, \maxR]\}$. A detailed description of this procedure can be found in Appendix~\ref{sec: findarea}.

We next show that the algorithm correctly computes the set $S(r)$, for every log-radius $r$. Define $\children(S(r))= \bigcup_{\tuple{j,r} \in S(r)} \children(j,r)$.

\begin{lemma} \label{lem: look-up-balls-correct}
 For each logradius $r \in [\maxR, \minR)$ assume $S(r)$ is computed correctly. Then it holds that $S(r-1) \subseteq \children(S(r))$.
\end{lemma}

\begin{proof} Assume towards contradiction that there exists $\tuple{j,r-1} \in S(r-1)$ such that $\dist(j,p) \leq c^* \cdot 5^{r-1}$ but $\tuple{j,r-1} \not \in \children(S(r))$. 
Let $\tuple{j',r}$ be the parent of $\tuple{j,r-1}$ in $\mathcal{T}$. By Lemma~\ref{lemma:inclusion-of-big-balls_2}, $\dist(j,j') \leq c_1 \cdot 5^{r}$. 

Now, since $S(r)$ is correct, it follows that $\tuple{j',r} \not \in S(r)$, and thus $\dist(p,j') > c^* \cdot 5^{r}$. However, by Lemma~\ref{lemma:inclusion-of-big-balls_2} we get that $p \in B(j,c^* \cdot 5^{r-1}) \subseteq B(j',c^* \cdot 5^{r})$ and thus $\dist(p,j') \leq c^* \cdot 5^{r}$, which is a contradiction. Thus the lemma follows.
\end{proof}
Since $c_X \geq (5/4)c_1$, let $c^* = c_X$. We now argue about the running time of \textsc{FindBalls}$(p,c^*)$. Note that for each logradius $r$, if $p \in B(j,c^* \cdot 5^{r})$, then $\dist(p,j) \leq c^* \cdot 5^{r}$. By Lemma~\ref{lem:doubling_dimension_2} and the fact that $c^* \leq 4 c_1$ it follows that $\abs{S(r)} \leq 2^{3\kappa}$. Additionally, by Lemma~\ref{lemma:degree_2} each pair in $S(r)$ has at most $2^{4\kappa}$ children in $\mathcal{T}$ and there are at most $\Delta$ different radii. Thus the running time of the algorithm and the size of the output set $S$ are both bounded by $2^{7 \kappa} \cdot \Delta$. 

\begin{lemma} \label{lem: findBalls}
The running time of $\textsc{FindBalls}(p,c^*)$ and the size of the output set $S$ are both bounded by $2^{7 \kappa} \cdot \Delta$.
\end{lemma}
Repeatedly applying the $\textsc{FindBalls}$ subroutine and updating the tree hierarchy in a bottom-up fashion, we can show that insertions and deletions of clients can be handled in ${O}(2^{O(\kappa^2)} \cdot \Delta^3 \cdot (\kappa^2 + \log \Delta))$ time. Additionally note that under client updates, the value of $n$ will change, which in turn causes $\minR$ to either increase or decrease by one. This forces us to either add or delete a bottom-level in the hierarchy, which can be implemented in ${O}((\abs{J} + \abs{\Clients})\cdot 2^{O(\kappa^2)} \cdot \Delta^3 \cdot (\kappa^2 + \log \Delta))$ time. Since such an update is required only after $\Theta(n)$ operations, we get that the amortized time of our algorithm is still bounded by ${O}(2^{O(\kappa^2)} \cdot \Delta^3 \cdot (\kappa^2 + \log \Delta))$.  By employing a standard global rebuilding technique we achieve a worst-case update time, thus proving our main result in Theorem~\ref{thm: mainTheorem}. Details on implementing the above steps can be found in Appendices~\ref{lem: updateClients} and~\ref{sec: insertingLevels}.

\bibliography{literature}

\appendix
\section{Missing Proofs from Section~\ref{sec:2}}

\subsection{Hierarchical decomposition of $\Pairs$}

\begin{proof}[Proof of Lemma~\ref{lemma:inclusion-of-big-balls_2}]
We have $\tuple{j,r}\in \Pairs$, so $j \in \Facilities$. 
By the Covering property of $\Pairs$ there exists  $\tuple{j',r+1}\in \Pairs$ such that $\dist(j,j')\leq c_1\cdot 5^{r+1}$. Since $c^*\leq (5/4)c_1$, we have $c_1\cdot 5^{r+1}+c^*\cdot 5^r \leq c^*\cdot 5^{r+1}$ and the result holds.  
\end{proof}

\begin{proof}[Proof of Lemma~\ref{lem:doubling_dimension_2}]
By definition of doubling dimension, $B(p, 2^{\alpha}c_1\cdot 5^r)$ can be
covered by a set of at most $(2^{\kappa})^{\alpha+1}$ balls of radius $(c_1/2)\cdot 5^r $.
By the Separating property of $\Pairs$, any two pairs $\tuple{j,r}$ of $\Pairs(p,r)$ are at distance at least $c_1\cdot 5^r$ from each other, hence must belong to  different balls of the set, and so $\Pairs(p,r)$ has cardinality at most $(2^{\kappa})^{\alpha+1}$. 
\end{proof}

\begin{proof}[Proof of Lemma~\ref{lemma:degree_2}]
Children of node $\tuple{j,r}$
have logradius $r-1$, so by the Separating property of $\Pairs$ they have to be at least $c_1\cdot5^{r-1}$ apart. Let $c^* = c_2$, since $c_2 \leq (5/4)c_1$.
Then, by Lemma~\ref{lemma:inclusion-of-big-balls_2} they must be at distance at most $c_2\cdot 5^r-c_2\cdot 5^{r-1}=(4/5)c_2\cdot 5^r$ from $j$. 
Since $(4/5)c_2 \cdot 5^{r} = 7 c_1 \cdot 5^{r-1}$, by Lemma~\ref{lem:doubling_dimension_2} for $\alpha=3$ and logradius $r-1$, the size of that set is bounded by $2^{4\kappa}$.
\end{proof}

\subsection{Hierarchical decomposition of $V$ into a laminar family of areas}

\begin{proof}[Proof of Lemma~\ref{lemma:area-contained-in-big-ball_2}]
Let $p\in A(j,r)$. Either it's been added directly, in which case it belongs to $B(j,c_2\cdot 5^r)$, or it's been inherited, in which case it also belongs to it by Lemma~\ref{lemma:inclusion-of-big-balls_2}.
\end{proof}

\begin{proof}[Proof of Lemma~\ref{lemma:covering-balls-by-areas_2}]
By the Covering property of $\Pairs$ it follows that there exists $\tuple{j_r,r} \in \Pairs$ such that $\dist(j,j_r) \leq c_1 \cdot 5^{r}$. Then $\dist(p,j_r) \leq (c_1 + 1)\cdot 5^{r} < c_2 \cdot 5^{r}$ since $c_1 < c_2 - 1$ and so in the definition of areas covering $p$ we must have $r^{\ropt} \leq r$. Along the path from $\tuple{j^{\ropt},r^{\ropt}}$ to the root of $\mathcal{T}$, there exists a pair for logradius $r$, $\tuple{j',r}$. By definition of areas and Lemma~\ref{lemma:area-contained-in-big-ball_2}, we have $p \in A(j',r) \subseteq B(j',c_2\cdot 5^{r})$, so $\dist(j,j') \leq \dist(j,p) + \dist(p,j') \leq (1+c_2)\cdot 5^{r}$.
\end{proof}

\begin{proof}[Proof of Lemma~\ref{lem:areas-equal-balls_2}]
By Lemma~\ref{lemma:area-contained-in-big-ball_2} it follows that
$\bigcup_{\tuple{j,r} \in \Facilities_r}A(j,r) \subseteq \bigcup_{\tuple{j,r} \in \Facilities_r}B(j,c_2\cdot 5^r).$

On the other side the definition of areas implies that any point $p \in B(j,c_2\cdot 5^r)$ for some $\tuple{j,r} \in  \Pairs$ will be added to an area.
Thus,
$\bigcup_{\tuple{j,r} \in \Facilities_r}A(j,r) \supseteq \bigcup_{\tuple{j,r} \in \Facilities_r}B(j,c_2\cdot 5^r).$
The lemma follows.
\end{proof}

\subsection{Covering balls using areas}


\begin{proof}[Proof of Lemma~\ref{lem: x-cont-ball}]
Let $p \in \X(j,r)$. Then there exists an area $A(j',r)$ such that $p \in A(j',r)$ and $\dist(j,j') \leq c_X \cdot 5^{r}$. By Lemma~\ref{lemma:area-contained-in-big-ball_2}, $A(j',r) \subset B(j',c_2 \cdot 5^{r})$ and thus $\dist(p,j') \leq c_2 \cdot 5^{r}$. It follows that $\dist(p,j) \leq \dist(p,j') + \dist(j',j) \leq c_3 \cdot 5^{r}$, i.e., $p \in B(j,c_3\cdot 5^{r})$.
\end{proof}

\begin{proof}[Proof of Lemma~\ref{lem: y-of-ancestor}]
It is enough to prove the lemma for $r' = r+1$, that is, when $\tuple{j',r'}$ is the parent of $\tuple{j,r}$ in $\T$.
Lemma~\ref{lemma:inclusion-of-big-balls_2} implies that $\dist(j,j') \leq c_1\cdot5^{r+1}$.
Let $A(j'',r) \in \Y(j,r)$, in particular $\dist(j'',j) \leq c_Y\cdot 5^r$.
Denote by $\tuple{j''', r+1}$ the parent of $\tuple{j'',r}$ in $\T$ and observe that
\[
\dist(j',j''') 
  \leq \dist(j',j) + \dist(j,j'') + \dist(j'',j''') 
  \leq c_1\cdot5^{r+1} + c_Y\cdot 5^r + c_1\cdot 5^{r+1} 
  \leq c_Y\cdot 5^{r+1}.
\]
The latter inequality holds since $10 c_1 \le 4 c_Y$.
This means $A(j''',r+1) \in \Y(j',r+1)$ and together with $A(j'',r) \subseteq A(j''', r+1)$ concludes the proof.
\end{proof}

\section{Missing Proofs from Section~\ref{sec:3}}

\subsection{Open triplets}

\begin{proof}[Proof of Lemma~\ref{lem:oneOpen2}]
First note that by the definition of $\maxR$, $5^{\maxR} \ge \fmax$ and $5^{\maxR} \ge W$. Let $\tuple{j^{\rroot}, \maxR, s}$ be the root of $\T$. Note that $A(j^{\rroot}, \maxR)= \X(j^{\rroot}, \maxR) = V$ and thus the color $s$ is irrelevant. Since by assumption $\abs{\Clients} > 0$, it follows that $\abs{\Clients \cap \X(j^{\rroot},\maxR)} \ge 1 \ge f^*_{\tuple{j^{\rroot},\maxR}}/5^{\maxR}$. The latter implies that $\tuple{j^{\rroot}, \maxR, s}$ fulfills the abundance condition and either $\tuple{j^{\rroot}, \maxR, s}$ is open or there exists another triplet of smaller log-radius that is open.
\end{proof}

\begin{proof}[Proof of Lemma~\ref{lem:at_most_one_ball}]
Assume towards contradiction that $i$ belongs to $\X(j,r)$ and $\X(j', r')$ with
$\tuple{j,r, s} \in J^\ropen(\Clients)$ for some color $s$ and $\tuple{j',r', s'} \in J^\ropen(\Clients)$ for some color $s'$. 
W.l.o.g $r \ge r'$. 

{\bf Case 1:} $r = r'$.
 By Lemma~\ref{lem: x-cont-ball}, it follows that $\dist(j, j') \le 2c_3 \cdot  5^r$. But then 
by the definition of coloring it is impossible that $s = s'$. Assume w.l.o.g that $s' < s$. By the definition of $\Y(j,r)$ it follows that
$A(j', r') \subseteq \Y(j,r)$ since $\dist(j,j') \le c_Y \cdot 5^r$, and, thus, by Lemma~\ref{lemma:area-contained-in-big-ball_2}, $j' \in \Y(j,r)$.
But this is a contradiction to the fact that  $\tuple{j,r, s} \in J^\ropen(\Clients)$.

{\bf Case 2:} $r > r'$.
 By Lemma~\ref{lem: x-cont-ball}, it follows that $\dist(j, j') \le c_3 (5^r + 5^{r'})$.
Let $\tuple{j'', r}$ be the ancestor of $\tuple{j', r'}$ of logradius $r$. Then $\dist(j'', j') \le  \frac{5}{4}c_1 \cdot 5^r$,
implying that $\dist(j, j'') \le  c_3 (5^r + 5^{r'}) + \frac{5}{4}c_1 \cdot 5^r \le c_Y \cdot 5^r$.
It follows from  Lemma~\ref{lemma:area-contained-in-big-ball_2} that $j' \in A(j', r') \subseteq A(j,r) \subseteq Y(j,r)$.
But  this is a contradiction to $\tuple{j,r, s} \in J^\ropen(\Clients)$.
\end{proof}

\subsection{Enabled triplets}
\begin{proof}[Proof of Lemma~\ref{obs:abundance_implies_enabled_2}]
Consider two cases: either there exists some open triplet $\tuple{j',r',s'}$
such that $\tuple{r',s'} \leq_\rlex \tuple{r,s}$ and $j'$ falls into any area within $\Y(j,r)$, or not.
If it does, then $\tuple{j,r,s}$ is enabled by the definition of 
$J^\renbl_{r,s}(\Clients)$.
Otherwise, $\tuple{j,r,s}$ is open, and thus also enabled.
\end{proof}

\subsection{Assignment of Clients}

\begin{proof}[Proof of Observation~\ref{obs:again_nearby_open_facility}]
Recall that $i \in A(j^\rarea_i, r^\rarea_i)$.
Thus
\begin{align*}
 \dist(i,j^{\ropen}_i)            &\leq \dist(i,j^\rarea_i) 
                + \dist(j^\rarea_i, j^\rdesg_i)
 + \dist(j^\rdesg_i, j^\ropen_i) \\
             &\leq c_2 \cdot 5^{r^\rarea_i}
                + c_4 \cdot 5^{r^\rarea_i} +c_3 \cdot 5^{r_i^\rarea}\\
             &=(c_2 + c_4 + c_3) \cdot 5^{r^\rarea_i}.
\end{align*}
The last inequaility follows from Lemma~\ref{lemma:area-contained-in-big-ball_2}, the definition of enabled triplets, and Observation~\ref{obs:openf_2}.
\end{proof}

\begin{proof}[Proof of Lemma~\ref{lem:radiusMatch2}]
Assume towards contradiction that $r_{i}^{\rarea} < r$.
Then there exists an enabled triplet $\tuple{j', r', s'} \in J^\renbl_\Clients$ with $r' = r_{i}^{\rarea}$ and $i \in A(j', r')$.
By the definition of $J_{r',s'}^\renbl(\Clients)$ we know that 
there must be an open triplet $\tuple{j'',r'', s''}$ with $\tuple{r'',s''} \leq_\rlex \tuple{r',s'}$ 
such that $j''$ falls within $\Y(j',r')$. 
Let $\tuple{j''', r}$ be the ancestor of $\tuple{j'',r''}$ of logradius $r$ in $\T$.
Taking into account that $r' \leq (r-1)$ we can derive that
\begin{align*}
	\dist(j,j''') 
    &\leq \dist(j,i) + \dist(i,j') + \dist(j',j'') + \dist(j'', j''') \\
    &\leq c_3 \cdot 5^r + c_2 \cdot 5^{r'} + c_4 \cdot 5^{r'} + \frac{5}{4}c_1 \cdot 5^r\\
    &\leq 2c_3 \cdot 5^{r} \\
  & \leq c_Y \cdot 5^r,
\end{align*}
which implies that $A(j''',r) \subseteq \Y(j,r)$. By the nesting property of areas in Lemma~\ref{lemma:area-contained-in-big-ball_2} 
it follows that   $j'' \in A(j''',r)$ which implies that $j'' \in \Y(j,r)$. This contradicts the claim that $\tuple{j,r,s}$ is an open triplet.
\end{proof}

\begin{proof}[Proof of Lemma~\ref{lm:facilities_costs_still_dont_matter}]
The first part follows from Lemma~\ref{obs:again_nearby_open_facility}.
The second part follows from the abundance condition, namely
\begin{align*}
\sum_{j \in \mathcal{\Solution_{\Clients}}}f_j 
  &\leq \sum_{\tuple{j,r,s} \in J^\ropen_\Clients}f^*_{\tuple{j,r}} \leq \sum_{\tuple{j,r,s} \in J^\ropen_\Clients} 5^r 
     \cdot \abs{\Clients\cap \X(j,r)} =    \sum_{\tuple{j,r,s} \in J^\ropen_\Clients} 
        \sum_{i \in \Clients \cap \X(j,r)} 5^r \leq \sum_{i \in \Clients} 5^{r^\rarea_i}
\end{align*}
The last inequality follows from Lemma~\ref{lem:at_most_one_ball} and~\ref{lem:radiusMatch2}.
\end{proof}


\begin{figure} 
\begin{center}
\fbox{\parbox{\textwidth}{
Let $S(\maxR) \gets \{\tuple{j,\maxR}\}$, $S(r) \gets \emptyset$, $\forall r < \maxR$.
\begin{itemize}
\setlength\itemsep{-0.05cm}

\item
for all $r \in [\maxR, \minR)$ and $\tuple{j,r} \in S(r)$:
	\begin{itemize}
	\item for all $\tuple{j',r-1} \in \children(j,r)$:
	\begin{itemize}
	\item if $\dist(p,j') \leq c^* \cdot 5^{r}$, then $S(r-1) \gets S(r-1) \cup \{\tuple{j',r-1}\}$.
	\end{itemize}
	\end{itemize}
\end{itemize}
$S \gets \bigcup \{S(r) : r \in [\minR, \maxR]\}$.
}}
\end{center}
\caption{\textsc{FindBalls}$(p,c^*)$.}
\label{fig:findBalls}
\end{figure}



%

\subsection{Finding the area containing a given point} \label{sec: findarea}

\begin{figure} 
\begin{center}
\fbox{\parbox{\textwidth}{
$S \gets \textsc{FindBalls}(p,c_2)$.
\begin{itemize}
\item let $r^*$ be minimum such that there exists pairs $\tuple{j,r^*} \in S$.
\item among all such pairs, output the pair $\tuple{j^*,r^*}$ minimizing $\dist(p,j)$. 
\end{itemize}
}}
\end{center}
\caption{\textsc{FindArea}$(p)$.}
\label{fig:findArea}
\end{figure}

Recall the algorithm $\textsc{FindBalls}(p,c^*)$ from Section~\ref{sec:lookup_ball}~(as depicted in Figure~\ref{fig:findBalls}). In this section we show how to find the area containing a given point. Specifically, given an arbitrary point $p \in V$, our goal is to determine the pair $\tuple{j^*,r^*}$ with smallest logradius such that $p \in A(j^*,r^*)$. By Lemma~\ref{lemma:area-contained-in-big-ball_2}, $p \in B(j,c_2 \cdot 5^{r^*})$. Thus we will first find all balls $\tuple{j,r}$ such that $p \in B(j,c_2 \cdot 5^r)$, using the algorithm \textsc{FindBalls}$(p,c_2)$ (note that $c_2 \geq (5/4)c_1$, thus we set $c^* = c_2$). Then we search for $\tuple{j^*,r^*}$ in the output set $S$. The algorithm is summarized in Figure~\ref{fig:findArea}. 

Note that the running time of $\textsc{FindArea}(p)$ is dominated by the running time of $\textsc{FindBalls}(p,c_2)$, which similarly to Section~\ref{sec:lookup_ball} can be shown to be at most $2^{7 \kappa} \cdot \Delta$.
\begin{lemma} \label{lem:findArea}
Given a point $p \in V$, $\textsc{FindArea}(p)$ determines the pair $\tuple{j^{*},r^{*}}$ with smallest logradius such that $p \in A(j^*,r^*)$ and runs in $2^{7 \kappa} \cdot \Delta$ time. 
\end{lemma}

\section{Insertions and deletions of clients} \label{lem: updateClients}

\begin{figure}[t!]
\begin{center}
\fbox{\parbox{\textwidth}{ 
\begin{itemize}
\setlength\itemsep{-0.05cm}
\item Set $w \gets \textsc{FindArea}(p)$, where $w = \tuple{j',r',s'}$.
\item for all $u = \tuple{j,r,s} \in P(w)$:
\begin{itemize}
\item determine $\Pairs(u)$ using the algorithm $\textsc{FindBalls}(j,c_X)$.
\end{itemize}
\end{itemize}
$S \gets \bigcup \set{\Pairs(u)}{u \in P(w)}$.
}}
\end{center}
\caption{\textsc{FindAffectedTriplets}$(p)$.}
\label{fig:findAffectedTriplets}
\end{figure}

\begin{figure}[t!]
\begin{center}
\fbox{\parbox{\textwidth}{ 
\begin{itemize}
\item for all $v \in S$:
\begin{itemize}
\item set $n_{\rball}(v) \gets n_{\rball}(v) + 1$.
\item if $(5^{r} \cdot n_\rball(v) \geq f^{*}_{\tuple{j,r}})$, then $\isabundant^{\rnew}_{\rball}(v) \gets 1$, otherwise, $\isabundant^{\rnew}_{\rball}(v) \gets 0$.
\item if $\isabundant_{\rball}(v) \neq \isabundant^{\rnew}_{\rball}(v)$, then add $v$ to $\mathcal{H}$, and set $\isabundant_{\rball}(v) \gets \isabundant^{\rnew}_{\rball}(v)$.
\end{itemize}
\item while $\mathcal{H} \neq \emptyset$:
\begin{itemize}
\item pull out the minimum triplet $v$ from $\mathcal{H}$ and invoke \textsc{CheckStatus$(v)$}.
\item if ($\switched(v) = 1$, $\isopen_{\rball}(v) = 0$ and $\isopen^{\rnew}_{\rball}(v) = 1$), then
\begin{itemize}
\item set $\isopen_{\rball}(v) \gets \isopen^{\rnew}_{\rball}(v)$.
\item for all $u \in \neighborsabove(v)$:
\begin{itemize}
\item add $u$ to $\mathcal{H}$, 
 and set $\openbelow(u) \gets \openbelow(u) + 1$.
\end{itemize}
\end{itemize}
\item if ($\switched(v) = 1$, $\isopen_{\rball}(v) = 1$ and $\isopen^{\rnew}_{\rball}(v) = 0$), then
\begin{itemize}
\item set $\isopen_{\rball}(v) \gets \isopen^{\rnew}_{\rball}(v)$.
\item for all $u \in \neighborsabove(v)$:
\begin{itemize}
\item add $u$ to $\mathcal{H}$, and set $\openbelow(u) \gets \openbelow(u) - 1$.
\end{itemize}
\end{itemize}
\end{itemize}
\end{itemize}
Set $U \gets \emptyset$. 
\begin{itemize}
\item for all $v$ processed by $\mathcal{H}$:
\begin{itemize}
\item if $(\openbelow(v) \geq 1$ or $\isopen_{\rball}(v) = 1$), then $\isenabled^{\rnew}_{\rball}(v) \gets 1$, otherwise, $\isenabled^{\rnew}_{\rball}(v) \gets 0$.
\item if $\isenabled^{\rnew}_{\rball}(v) \neq \isenabled_{\rball}(v)$, then add $v$ to $U$. 
\end{itemize}
\end{itemize}
Return $U$.
}}
\end{center}
\caption{\textsc{UpdateStatus}$(S)$.}
\label{fig:updateStatus}
\end{figure}

\begin{figure}[t!]
\begin{center}
\fbox{\parbox{\textwidth}{ 
\begin{itemize}
\item if ($\isopen_{\rball}(v) = 0$, $\isabundant_{\rball}(v) =1$ and $\openbelow(v) = 0$) then
\begin{itemize}
\item set $\isopen^{\rnew}_{\rball}(v) \gets 1$ and $\switched(v) \gets 1$.
\end{itemize}
\item else if ($\isopen_{\rball}(v) = 1$ and $\openbelow(v) \geq 1$) or ($\isabundant_{\rball}(v) =0$) then
\begin{itemize}
\item set $\isopen^{\rnew}_{\rball}(v) \gets 0$ and $\switched(v) \gets 1$.
\end{itemize}
\end{itemize}
}}
\end{center}
\caption{\textsc{CheckStatus}$(v)$.}
\label{fig:checkStatus}
\end{figure}

In this section we describe how to update the annotated dependency tree $\mathcal{T}$ under insertions and deletions of clients. Here, we only focus on handling client insertions. Client deletions are handled in a symmetric way.

We divide the update algorithm into three parts. First, we determine all triplets that are affected, i.e., whose data structure needs to be updated. Second, we show how to update the status of effected pairs and resolve all the dependencies between possible status changes. Third, we deal with the cost computation. 

\vspace{0.2cm}
\noindent \textbf{Determining affected triplets.} When a client $p$ is inserted, we first find the triplet $w=\tuple{j',r',s'} \in \mathcal{T}$ such that $p \in A(j',r')$ and $r'$ is minimum, using the algorithm \textsc{FindArea}$(p)$ in Figure~\ref{fig:findArea}. Let $P(w)$ be the triplets in the (unique) path between $w$ and the root of $\mathcal{T}$. 
To update all counters $n_\rball(v),$ we need to determine all areas that are within distance $c_X \cdot 5^r$ for each triplet $u = \tuple{j,r,s} \in P(w)$, i.e.,
 the set $\Pairs(u)$ of all triplets $\tuple{j'',r,s''} \in J_r$ such that $\dist(j,j'') \leq c_X \cdot 5^{r}$. Since, $c_X \geq (5/4)\cdot c_1$, we use the algorithm $\textsc{FindBalls}(j,c_X)$ in Figure~\ref{fig:findBalls} to determine $\Pairs(u)$ and finally, let $S = \bigcup \set{\Pairs(u)}{u \in P(w)}$ be the set of affected triplets. This algorithm is summarized in Figure~\ref{fig:findAffectedTriplets}.

We next argue about the running time and the size of $S$. Note that $\abs{P(w)} \leq \Delta$. By Lemma~\ref{lem: findBalls}, $\textsc{FindBalls}$ runs in $O(2^{7\kappa} \cdot \Delta)$ time. In addition, by Lemma~\ref{lem:doubling_dimension_2} and the fact that $c_X \leq 4c_1$ it follows that $\abs{\Pairs(u)} \leq 2^{3\kappa}$. Combining the above bounds, we get that $\abs{S} \leq 2^{3\kappa} \cdot \Delta$ and the algorithm runs in $O(2^{7\kappa} \cdot \Delta^2)$ time.

\begin{lemma} \label{lem:affectedTriplets}
The running time of $\textsc{FindAffectedTriplets}(p)$ and the size of $S$ are bounded by $O(2^{7\kappa} \cdot \Delta^2)$ and $2^{3\kappa} \cdot \Delta$, respectively.
\end{lemma}
\vspace{0.2cm}

\noindent \textbf{Updating status of the affected triplets and their $n_\rball$ value.} We start by introducing some useful notation. We say that a triplet $\tuple{j,r,s}$ is \emph{dirty} iff its status switched from open to closed, from closed to open, or if it is an above-neighbor of such a triplet. To deal with triplets that change their status in an \emph{efficient} way, the algorithm uses a min-heap $\mathcal{H}$, which is maintained using the lexicographic ordering on $\tuple{r,s}$.

Let $S$ denote the set of affected triplets, output by $\textsc{FindAffectedTriplets}(p)$. For each $v = \tuple{j,r,s} \in S$,  we increment $n_{\rball}(v)$ and check whether $v$ fulfils the abundance condition (note that $n_{\rarea}$ counters are not yet incremented and we will deal with them later). If a triplet newly fulfils the abundance condition, or no longer fulfils this condition, we add it to the min-heap $\mathcal{H}$. Now, as long as $\mathcal{H}$ is not empty, we pull the minimum triplet $v$ from $\mathcal{H}$, and check whether $v$ has to change its status, using the algorithm \textsc{CheckStatus}$(v)$ in Figure~\ref{fig:checkStatus}, in a way to be described shortly. If $v$ just switched to open, then we need to notify all the above neighbors about this change. Specifically, we consider all above neighbors of $v$ in turn, add them to $\mathcal{H}$ and increment their $\openbelow$ counter. Symmetrically, if $v$ just switched to closed, we add all above neighbors of $v$ to $\mathcal{H}$ and decrement their $\openbelow$ counter. Finally, after $\mathcal{H}$ becomes empty, we can check using the $\openbelow$ counter whether the triplets that were added and removed from $\mathcal{H}$ switched their status to enabled or vice versa. This algorithm is summarized in Figure~\ref{fig:updateStatus}.

Given a triplet $v$, the algorithm \textsc{CheckStatus}$(v)$ tests whether $v$ has to change its status. If $v$ is currently closed, fulfils the abundance condition and no below neighbor is open, then we make $v$ open. Otherwise, if it is currently open and at least one below neighbor is open, or abundance condition is no longer fulfilled, then we make $v$ closed.

We now show the correctness of the update algorithm. We say that a dirty triplet $v$ is \emph{cleaned} if it was pulled out the min-heap $\mathcal{H}$ and $\textsc{CheckStatus}(v)$ was invoked in Figure~\ref{fig:updateStatus}.

\begin{lemma} If a dirty triplet $v$ is cleaned, then the status of $v$ is correct.
\end{lemma}  
\begin{proof}
We first deal with the open/closed status. A triplet $v$ is added to $\mathcal{H}$, i.e., it is dirty, iff there is a change in the abundance condition of $v$, and/or if it is an above neighbor of some dirty node. 

\vspace{0.2cm}
\noindent ~(1) Suppose there was a change in the abundance condition of $v$. We distinguish two cases:
\begin{itemize}
\item Triplet $v$ \emph{newly} satisifes the abundance condition. Then $v$ is currently closed. Now, either there is some below neighbor that is open, or this is not the case. In the first scenario, the status of $v$ remains unchanged, while in the second the algorithm correctly switches its status to open. 
\item Triplet $v$ \emph{no longer} satsifies the abundance condition. Then $v$ is currently either closed or open. In the first scenario, the status of $v$ remains unchanged, while in the second  the algorithm correctly switches its status to closed. 
\end{itemize}

\noindent ~(2) Suppose $v$ is an above neighbor of some dirty node $u$. We distinguish two cases:
\begin{itemize}
\item Triplet $u$ switched to \emph{open} and we incremented the $\openbelow(v)$ counter by $1$. Then $v$ is currently either closed or open. In the first scenario, the status of $v$ remains unchanged, while in the second the algorithm correctly switches its status to closed.
\item Triplet $u$ switched to \emph{closed} and we decremented the $\openbelow(v)$ counter by $1$. Then $v$ is currently closed. Now, either there is some below neighbor that is open, or this is not the case. In the first scenario, the status of $v$ remains unchanged, while in the second the algorithm correctly switches its status to open. 
\end{itemize}

Now, a triplet $v$ enabled if there is some below neighbor that is open.  The algorithm precisely checks this using the counter $\openbelow(v)$ and sets its status accordingly. 
\end{proof}

The following lemma is crucial for the correctness of our algorithm and for bounding the overall number of dirty triplets. 

\begin{lemma} 
A dirty triplet once cleaned cannot again become dirty.
\end{lemma}
\begin{proof}
Let $v = \tuple{j,r,s}$ be a dirty triplet that is cleaned. Assume towards contradiction that $v$ is again added to $\mathcal{H}$, i.e., it is dirty. Then it must have been added because it is an above neighbor of some triplet $u = \tuple{j',r',s'}$ such that $\tuple{r',s'} <_{\rlex} \tuple{r,s}$. But since $v$ is cleaned and the min-heap $\mathcal{H}$ is maintained according to lexicographic ordering on pairs, we get that all dirty triplets with $\tuple{r',s'} <_{\rlex} \tuple{r,s}$ are also cleaned and they are processed before $v$. 
Additionally, all triplets that are added to $\mathcal{H}$ during the while-loop appear later in the lexicographic order than $v$, as they are either above-neighbors of $v$ or of a triplet that comes later than $v$ in the lexicographic order.  It follows that $u$ could not be processed at this time step and thus, $v$ is not added to $\mathcal{H}$, which is a contradiction. 
\end{proof}

Note that if no triplet changes its status from open to close or vice versa, no further changes are necessary. However, if a triplet changes its status we need to show how to bound the number of triplets changing their status. 

\begin{lemma} \label{lem: numberChanges}
If some triplet $\tuple{j,r,s}$ becomes open (resp.~closed), then at most 
$2^{5\kappa^2 + 6\kappa}\cdot\Delta$ other triplets
become closed (resp.~open). 
\end{lemma}
\begin{proof}
Assume that the triplet $v = \tuple{j,r,s}$ changes its status. A triplet 
$u = \tuple{j',r',s'}$ is {\em directly influenced} by $v$ if
  $\tuple{r,s}<_\rlex\tuple{r',s'}$, $j \in \Y(j',r')$, and $u$ fulfills the abundance condition. Additionally, a triplet $u = \tuple{j',r',s'}$ is {\em indirectly influenced} if (a) there exist triplets
$u_1, \dots,  u_\ell$ with $\tuple{r,s} <_{\rlex} \tuple{r_1,s_1} <_{\rlex} \ldots <_{\rlex} \tuple{r_{\ell},s_{\ell}} <_{\rlex} \tuple{r',s'} $, (b)
triplet $u_{\ell}$ is directly or indirectly influenced by
$v$ and (c) triplet $u$ is directly influenced by $u_{\ell}$. 

For any logradius $r' > r$, let $\tuple{j', r', s'}$ be the triplet directly or indirectly influenced by $v$ with maximum distance between $j$ and $j'$.
Note that it is possible that a triplet $\tuple{j',r,s'}$ is directly or indirectly influenced by a triplet $\tuple{j,r,s}$ with the same logradius. However, in that case  $s' > s$, which limits the number of triplets with logradius $r$ to $2^{5\kappa}+1$ in any sequence of indirectly influenced triplets.
It follows that
\[
\dist(j,j') \leq \sum_{j = r+1}^{r'}  c_4\cdot5^{j} \cdot (2^{5\kappa}+1)
            \leq (5/4) \cdot c_4 \cdot (2^{5\kappa}+1) \cdot 5^{r'} 
            <  16 \cdot (2^{5\kappa}+2^{5\kappa}) \cdot c_1 \cdot 5^{r'}.
\]
Now, by Lemma~\ref{lem:doubling_dimension_2} we get that there are 
at most $2^{(5\kappa+6)\kappa}$ 
facilities of radius $r'$ within the ball $B(j,\dist(j,j'))$, and hence at most $2^{5\kappa^2 + 6\kappa} $ triplets of log-radius $r'$ indirectly influenced by triplet $v = \tuple{j,r,s}$.
As there are $\Delta$ different possible radii, the theorem follows.
\end{proof}

We now argue about the running time of $\Call{UpdateStatus}{S}$. First, by Lemma~\ref{lem:affectedTriplets} recall that the time to find $S$ and its size are bounded by $O(2^{7\kappa} \cdot \Delta^2)$ and $O(2^{7\kappa} \cdot \Delta)$, respectively.
Next, note that the heap $\mathcal{H}$ initially contains only elements of $S$. For each $v \in S$,  Lemma~\ref{lem: numberChanges} implies that there are at most $2^{5\kappa^2+6\kappa} \cdot \Delta$ other triplets that are affected by $v$, and thus added to $\mathcal{H}$. It follows that the size of $\mathcal{H}$ is at most $2^{5\kappa^{2} + 13\kappa} \cdot \Delta^2$. When a triplet $v$ is pulled out from $\mathcal{H}$, in the worst-case the algorithm iterates over all above neighbors of $v$. Since there can be at most $2^{6\kappa} \cdot \Delta$ such neighbors, we get that the number of heap operations needed for processing all triplets from $\mathcal{H}$ is bounded by $2^{5\kappa^{2} + 19\kappa} \cdot \Delta^3$. Each heap operation takes $O(\log \abs{\mathcal{H}}) = O(\kappa^{2} + \log \Delta)$ time, so the running time for the \textsc{while} loop is bounded by $O(2^{5\kappa^{2} + 19\kappa} \cdot \Delta^3 \cdot (\kappa^{2} + \log \Delta))$. Finally, observe that $\abs{U} \leq \abs{\mathcal{H}}$ and thus the last \textsc{for} loop requires $O(2^{5\kappa^{2} + 13\kappa} \cdot \Delta^2)$ time. Combining the above bounds, the total time of the procedure is bounded by $O(2^{O(\kappa^{2})} \cdot \Delta^3 \cdot (\kappa^{2} + \log \Delta)).$

\begin{lemma} \label{lem: run-time-update}
The running time of $\Call{UpdateStatus}{S}$ is bounded by $O(2^{O(\kappa^{2})} \cdot \Delta^3 \cdot (\kappa^{2} + \log \Delta))$.
\end{lemma}

\begin{figure}[t!]
\begin{center}
\fbox{\parbox{\textwidth}{
Let $U \gets \textsc{UpdateStatus}(S)$.
\begin{itemize}
\item for all $u \in U$:
\begin{itemize}
\item for all $v \in P(u)$:
\begin{itemize}
\item set $n_{\enabledbelow}(\parent(v))  \gets n_{\enabledbelow}(\parent(v))$
\vspace{0.1cm}
\item[] \hspace{110pt}$+~n_{\rarea}(v) \cdot(\isenabled_{\rball}^{\rnew}(v) - \isenabled_{\rball}(v))$.
\end{itemize}
\vspace{0.2cm}
\end{itemize}
\item Set $w \gets \textsc{FindArea}(p)$, where $w = \tuple{j',r',s'}$.
\item for all $v \in P(w)$:
\begin{itemize}
\item set $n^\rnew_\rarea(v) \gets n_\rarea(v)+1$.
\item set $n_{\enabledbelow}(\parent(v))  \gets n_{\enabledbelow}(\parent(v))$
\vspace{0.1cm}
\item[] \hspace{110pt}$+~n^{\rnew}_{\rarea}(v) \cdot \isenabled_{\rball}^{\rnew}(v) - n_{\rarea}(v) \cdot \isenabled_{\rball}(v)$. 
\item set $n_{\rarea}(v) \gets n^{\rnew}_{\rarea}(v)$.
\end{itemize}
\vspace{0.2cm}
\item Let $Q \gets \bigcup \set{P(v)}{v \in U}$. 
\item Sort $Q$ in non-decreasing order according to radii of triplets.
\item for all $v \in Q$:
\begin{itemize}
\item  set $\cost(v) \gets y(v)  + (n_{\rarea}(v) - n_{\enabledbelow}(v)) \cdot \isenabled_{\rball}^{\rnew}(v) \cdot 5^{r}$.
\item update $y(\parent(v))$ according to the new value of $c(v)$.
\end{itemize}
\end{itemize}
}}
\end{center}

\caption{\textsc{UpdateCost}$(p)$}
\label{fig:updateCost}
\end{figure}
\vspace{0.2cm}

\noindent \textbf{Updating $n_{\rarea}$, $n_{\rball}$, $\cost$, and $y$.} Consider the set $U$ of output triplets by algorithm \textsc{UpdateStatus$(S)$}. Since these triplets have changed their status, by Lemma~\ref{lem:ancesEnabled_2}, we also need to consider the triplets belonging to the path between each triplet of $U$ to the root of $\mathcal{T}$. For any $u \in U$, let $P(u)$ denote such path. Now we iterate over all triplets from $\bigcup \set{P(u)}{u \in U}$ and update the counter $n_{\enabledbelow}$ of their parents. Note that these updates are sufficient for all triplets, except for the triplets whose corresponding areas contain the newly inserted client $p$. To fix this, we need to correctly set the counter $n_{\rarea}$ for such triplets.  


To this end, we find the triplet $w = \tuple{j',r',s'}$ in $\mathcal{T}$ such that $p \in A(j',r')$ and $r'$ is minimum, using the algorithm \textsc{FindArea}$(p)$ in Figure~\ref{fig:findArea}. Now, for each node $v = \tuple{j,r,s}$ along the (unique) path from $w$ to the root of $\mathcal{T}$ we perform the following: we increment $n_{\rarea}(v)$ and depending whether the area has switched its status, we also propagate this information to the parent of $v$, i.e., we update $n_{\enabledbelow}(\parent(v))$. Before proceeding further, note that the total cost that we pay for the clients in the subtree $\mathcal{T}(v)$, i.e., $\cost(v)$ with $v=\tuple{j,r,s}$, can be written recursively as follows: 

\begin{equation} \label{eq: costComput}
 c(v) =
  \begin{cases}
    y(v) + (n_{\rarea}(v)-n_{\enabledbelow}(v))\cdot 5^{r}       & \quad \text{if } v \text{ enabled}\\
    y(v)  & \quad \text{if } v \text{ not enabled}\\
  \end{cases}
\end{equation}

Now we show how to update the costs of the affected triplets, which we denote by $Q := \bigcup \set{P(v)}{v \in U}$. We first sort the triplets of $Q$ in non-decreasing order according to their radii. This is important as it ensures that cost updates are performed bottom-up, which in turn implies that the value $y(v)$ at each node is correct, when the re-computation of cost at that node is performed. Next, for each $v \in Q$, the cost of the solution $\cost(v)$ is computed using Equation~\eqref{eq: costComput} and the correctly updated values of $n_{\rarea}(v)$ and $n_{\enabledbelow}(\parent(v))$ from the previous steps. Finally, we update the cost of the parent of $v$, $y(\parent(v))$, according to the new value of $\cost(v)$. A detailed implementation is given in Figure~\ref{fig:updateCost}.



We now argue about the running time of $\Call{UpdateCost}{p}$. First, recall that $\abs{U} \leq 2^{5\kappa^2 + 19\kappa} \cdot \Delta^2$ and $\abs{P(u)} \leq \Delta$, for each $u \in U$. Thus, the time for updating the counters $n_{\enabledbelow}$ of the triplets belonging to $U$ is bounded by $O(2^{5\kappa^2 + 19\kappa} \cdot \Delta^3)$. Next, by Lemma~\ref{lem:findArea}, the running time for determining $w$ is bounded by $O(2^{7\kappa} \cdot \Delta)$. The updates involving the triplets in $P(w)$ take time proportional to the height of the tree, i.e., $O(\Delta)$. Finally, the above bounds together with the fact that $\abs{Q} \leq O(2^{5\kappa^2 + 19\kappa} \cdot \Delta^3)$ imply that the total time of the procedure is bounded by $O(2^{O(\kappa^2)} \cdot \Delta^3)$.

\begin{lemma} \label{lem:updateCost}
The running time of $\Call{UpdateCost}{p}$ is bounded by $O(2^{O(\kappa^2)} \cdot \Delta^3)$.
\end{lemma}
\vspace{0.2cm}

\noindent \textbf{Handling Queries.} For answering a cost query we return the cost of the root $\tuple{j,\maxR,s}$, i.e., $\cost(j,\maxR,s)$. Thus cost queries can be answered in $O(1)$ time. 
To handle a solution query we additionally keep a list of open facilities and for each open facility a pointer to its location in the list. Thus, facilities can be added or removed in constant time and
the set of open facilities can be output in time linear in its size.

\section{Insertions and deletions of levels in the hierarchy} \label{sec: insertingLevels}

Recall that $n$ is the largest power of $5$ smaller than the number of clients $\abs{\Clients}$. 
Suppose that the number of clients is a factor $5$ larger, resp. smaller, than $n$, due to insertions or deletions of clients. 
Then we update $n$ by multiplying, resp. diving, it by $5$. 
Furthermore, if $\abs{\Facilities} \leq n$ then  the value of $\minR$ has either decreased or increased by one, so an update of our data structure is required. 
Below we argue that the time for inserting, resp. deleting, a new, resp. old, level in the hierarchy is bounded by ${O}((\abs{J} + \abs{\Clients}) \cdot 2^{O(\kappa^2)} \cdot \Delta^3 \cdot (\kappa^2 + \log \Delta)) = \tilde{O}(n \cdot 2^{O(\kappa^2)})$. 
Since such an update is required only after $\Theta(n)$ operations, it follows that the \emph{amortized} running time of our algorithm is still bounded by ${O}(2^{O(\kappa^2)} \cdot \Delta^3 \cdot (\kappa^2 + \log \Delta)$. A standard global rebuilding technique can be employed to achieve a worst-case update time. 
In what follows, let $r = \minR$.

\subsection{Inserting a level in the hierarchy.}  For the new logradius $(r-1)$, we need to construct the maximal subset of distant facility/logradius pairs $J_{r-1}$, and update the hierarchical decomposition of $\Pairs$ by inserting another level at the bottom of the corresponding tree. Our approach is as follows.

We start by taking any pair from $J_r$ and adding it to $J_{r-1}$. Naturally, for each such pair we set $\parent(j,r-1) = \tuple{j,r}$, thus assigning the first dependencies in the extended tree. Next, for each facility $j \in J \setminus J_r$ we perform the following: first, we find all pairs $\tuple{j',r'}$ in $\Pairs$ such that $\dist(j,j') \leq c_2 \cdot 5^{r'}$, using the algorithm $\textsc{FindBalls}(j,c_2)$ in Figure~\ref{fig:findBalls}. Let $S$ denote the output set of pairs, and let $S(r)$, resp. $S(r-1)$, be a subset of $S$ containing pairs that share the same logradius $r$, resp. $(r-1)$. Next, define $\tuple{j^*,r}$ to be the closest pair to $j$ in $S(r)$. Finally, we check whether there is a pair $\tuple{j'',r-1} \in S(r-1)$ such that $\dist(j,j'') \leq c_1 \cdot 5^{r-1}$. If there is one, then we know that $\tuple{j,r-1}$ is covered and thus we do \emph{not} add it to $J_{r-1}$. Otherwise, we have that $\dist(j,j'') > c_1\cdot 5^{r-1}$, for all $j'' \in S(r-1)$, so we add $\tuple{j,r-1}$ to $J_{r-1}$ and set $\parent(j,r-1) = \tuple{j^*,r}$.

We now show the correctness. Note that setting $J_{r-1} =J_{r}$ is correct, since for any two pairs $\tuple{j,r}, \tuple{j',r} \in J_r$, $\dist(j,j') \geq c_1 \cdot 5^{r} \geq c_1 \cdot 5^{r-1}$. Now, let $j \in J \setminus J_r$. By definition, $S(r-1)= J_{r-1} \cap B(j,c_2 \cdot 5^{r-1})$ and since $c_1 < c_2$, it follows $J_{r-1} \cap B(j,c_1 \cdot 5^{r-1}) \subseteq S(r-1)$. Thus $S(r-1)$ contains all the relevant pairs which determine whether or not $\tuple{j,r-1}$ should be added to $J_{r-1}$. Following the construction of the hierarchy, it remains to argue that $\tuple{j^*,r}$ is the closest pair to $j$ in $J_r$. By the Covering Property of $\Pairs$, $\dist(j,j^*) \leq c_1 \cdot 5^{r} < c_2 \cdot 5^{r}$. Hence, by definition of $S(r)$, searching for $\tuple{j^{*},r}$ in $S(r)$ is correct. 

The running time of the above procedure is bounded by at most $\abs{J}$ calls to the $\textsc{FindBalls}$ algorithm, which in turn can by implemented in $2^{O(\kappa)}\cdot \Delta$ time, by Lemma~\ref{lem: findBalls}. This follows since $\abs{S(r)}$ and $\abs{S(r-1)}$ are both bounded by $2^{O(\kappa)}$, by the discussion in Section~\ref{sec:lookup_ball}, and thus the manipulations with these sets only asymptotically affect the running time. Hence, the running time for constructing $J_{r-1}$ and adding a level in the hierarchy is bounded by $\abs{J} \cdot 2^{O(\kappa)} \cdot \Delta = \tilde{O}(\abs{J} \cdot 2^{O(\kappa)})$. 
\vspace{0.2cm}

\noindent \textbf{Coloring of pairs, computing $\neighborsabove$ and designated facilities. } Let $J_{r-1}$ be the newly constructed set of pairs. Our goal is to efficiently color pairs of $J_{r-1}$ such that no two pairs $\tuple{j,r-1}, \tuple{j',r-1} \in J_{r-1}$ with $\dist(j,j') \leq c_4 \cdot 5^{r-1}$ share the same color. To this end, for any $\tuple{j,r-1} \in J_{r-1}$, let $\Pairs(j,r-1) = J_{r-1} \cap B(j,c_4 \cdot 5^{r-1})$. Define the graph $G=(V,E)$, where $V := J_{r-1}$, and $(\tuple{j,r-1},\tuple{j',r-1}) \in E$ iff $\tuple{j',r-1} \in \Pairs(j,r-1)$ or $\tuple{j,r-1} \in \Pairs(j',r-1)$. Then greedily coloring $G$ yields our desired coloring of the pairs in $J_{r-1}$. 

We now analyse the running time. Since $c_4 \geq (4/5)c_1$, for each $\tuple{j,r-1} \in J_{r-1}$ we can determine $\Pairs(j,r-1)$ by using the algorithm $\textsc{FindBalls}(j,c_4)$ in Figure~\ref{fig:findBalls}, whose running time is $2^{O(\kappa)} \cdot \Delta$. Observe that since maximum degree in $G$ is at most $2^{5\kappa} + 1$, the greedy colouring algorithm runs in time $O(\abs{J} \cdot 2^{O(\kappa)})$. Thus the running time of the coloring procedure is bounded by $\tilde{O}(\abs{J} \cdot 2^{O(\kappa)})$.

Given a triplet $v = \tuple{j,r-1,s} \in J_{r-1}$, recall that $\neighborsabove(v)$ is the list of all triplets $\tuple{j',r',s'}$ such that (a) $\tuple{r',s'} >_{\rlex} \tuple{r-1,s}$ and (b) $j \in \Y(j',r')$. Due to how $\Y(j',r')$ is defined, to check the second condition it is enough to test whether the logradius-$r'$ ancestor of $\tuple{j,r-1,s}$ belongs to $B(j', (c_2+2c_3)\cdot5^{r'})$. Similarly to the above, we can again use the algorithm \textsc{FindBalls} along with some post-preprocessing to compute these quantities and thus achieve a running time of $\tilde{O}(\abs{J} \cdot 2^{O(\kappa)})$.

We next show how to compute designated facilities and their cost. 
Recall that for any $\tuple{j,r-1,s} \in J_{r-1}$, the designated cost $f^{*}_{\tuple{j,r-1}}$ is the minimum opening cost among all facilities in $\X(j,r-1)$, i.e. $f^{*}_{\tuple{j,r-1}} = \min \set{f_{j'}}{j' \in J \cap \X(j,r-1)}$, and the designated facility is the one attaining the minimum cost. 
Our approach is as follows. 
First, for each $\tuple{j,r-1} \in J_{r-1}$, we maintain an initially empty min-heap $\mathcal{H}(j,r-1)$ that will store all pairs $(j',f_{j'})$, $j' \in A(j,r-1)$. 
Next, for each $j' \in J$, we use the algorithm $\textsc{FindArea}(j')$ in Figure~\ref{fig:findArea} to find the triplet $\tuple{j,r-1,s}$ such that $j' \in A(j,r-1)$ and 
insert $(j', f_{j'})$ into the heap $\mathcal{H}(j,r-1)$.
Note that each facility $j' \in J$ belongs to exactly one min-heap, since same-logradius areas are disjoint. 
Finally, for determining $f^{*}_{\tuple{j,r-1}}$ we proceed as follows: find all areas that are within distance $c_X \cdot 5^{r-1}$ from $j$ (i.e., all areas that belong to $\X(j,r-1)$), using the algorithm $\textsc{FindBalls}(j,c_X)$ in Figure~\ref{fig:findBalls}, and then set $f^{*}_{\tuple{j,r}} = \min\set{\mathcal{H}(j'r-1)\}}{A(j',r-1) \subseteq \X(j,r-1)}$. The designated facility is the
corresponding facility.

We now argue about the running time. Note that the size of each $\mathcal{H}(j,r-1)$ is bounded by $\abs{J}$. The above algorithm makes $\abs{J}$ calls to $\textsc{FindArea}$ and executes $\abs{J}$ heap operations. By Lemma~\ref{lem:findArea} and since each heap operation takes $O(\log \abs{J})$ time, it follows that the running time for this part is bounded by $O(\abs{J}\cdot 2^{O(\kappa)} \cdot \log \abs{J} \cdot \Delta)$. Next, we have $\abs{J_r} \leq \abs{J}$ calls to to $\textsc{FindBalls}$, which can be implemented in $2^{O(\kappa)}\cdot \Delta$, by Lemma~\ref{lem: findBalls}. Combining the above bounds, we get that the total running time is bounded by $\tilde{O}(\abs{J} \cdot 2^{O(\kappa)})$.

\vspace{0.2cm}

\noindent \textbf{Updating the annotated tree. } The three bits $\isopen_{\rball}(v)$, $\isenabled_{\rball}(v)$, $\isabundant_{\rball}(v)$ and the numbers $n_{\rarea}$, $y(v)$, $\cost(v)$ and $n_{\rball}(v)$ are initially set to $0$, for every $v = \tuple{j,r-1,s} \in J_{r-1}$. Additionally, note that $\openbelow$ and $\enabledbelow$ counters are also set to $0$, since $(r-1)$ is the last level of the tree. For determining $n_{\rarea}(v)$, resp. $n_{\rball}(v)$, i.e., the number of clients that belong to the area $A(j,r-1)$, resp. to $\X(j,r-1)$, and dealing with triplets whose status was affected, we follow an approach that is similar to the update algorithm for insertions/deletions of clients. Specifically, for each client $p \in \mathcal{C}$, we use the algorithm \textsc{FindArea}$(p)$ to find the triplet $w = \tuple{j',r',s}$ such that $p \in A(j',r')$ and $r'$ is minimum. If $r' = r-1$, then we increment the value of $n_{\rarea}(w)$. Next, for each $v \in J_{r-1}$, we find all areas that constitute $\X(j,r-1)$, i.e., areas that are within distance $c_X \cdot 5^{r-1}$ from $j$, using the algorithm $\textsc{FindBalls}(j,c_X)$. By definition of $\X(j,r-1)$, we get that $n_{\rball}(v) = \sum \set{n_{\rarea}(j',r-1,s')}{ A(j',r-1) \subseteq \X(j,r-1)}$. Finally, for each $v \in J_{r-1}$, we set $S=v$ and exactly proceed as in the algorithm $\textsc{UpdateStatus}(S)$ in Figure~\ref{fig:updateStatus}, except that we do not increment the $n_{\rball}$ counters.

We now analyse the running time. 
As above, the time to find the smallest logradius area $A(j',r')$ for each client, and then determine $\X(j,r)$ for each triplet in $J_{r-1}$, is bounded by $O((\abs{\Clients} + \abs{J}) \cdot 2^{O(\kappa)} \cdot \Delta)$. Furthermore, since we make at most $\abs{J}$ calls to the $\textsc{UpdateStatus}$, Lemma~\ref{lem: run-time-update} implies that the update time is at most $O(\abs{J} \cdot 2^{O(\kappa^2)} \cdot \Delta^3 \cdot (\kappa^2 + \log \Delta))$. Combining the above bounds, we get a total running time of ${O}((\abs{J} + \abs{\Clients}) \cdot 2^{O(\kappa^2)} \cdot \Delta^3 \cdot (\kappa^2 + \log \Delta))$.

Next we use the algorithm \textsc{UpdateCosts} in Figure~\ref{fig:updateCost}, to update the  counters $n_{\enabledbelow}$ and the costs of the triplets that changed their status,  but not performing the computations involving areas. By Lemma~\ref{lem:updateCost}, one can similarly argue that the running time for updating these information is also bounded by 
${O}(\abs{J} \cdot 2^{O(\kappa^2)} \cdot \Delta^3)$.
\vspace{0.2cm}

\noindent \textbf{Deleting a level from the hierarchy. }When deleting the last level $J_r$, for each triplet $\tuple{j,r,s} \in J_r$, we only need to check the above neighbors as their status might have been affected. Thus, following a similar approach as above, we can update the annotated tree and ensure correctness for the remaining levels in the hierarchy. The running time is bounded by $O(\abs{J} \cdot 2^{O(\kappa^2)})$.

\begin{lemma} 

Adding or deleting a bottom-level in the hierarchy can be implemented in ${O}((\abs{J} + \abs{\Clients}) \cdot 2^{O(\kappa^2)} \cdot \Delta^3 \cdot (\kappa^2 + \log \Delta))$ time.
\end{lemma}

The above lemma along with Lemmas~\ref{lem:affectedTriplets},~\ref{lem: run-time-update}, and~\ref{lem:updateCost} prove our main theorem:

\begin{theorem}[Restatment of Theorem~\ref{thm: mainTheorem}]
There exists a deterministic algorithm for the dynamic facility location problem when clients and facilities live in a metric space with doubling dimension $\kappa$, such that at every time step the solution has cost at most $O(1)$ times the cost of an optimal solution at that time. The worst-case update time for client insertion or deletion is $O(2^{O(\kappa^2)} \cdot \Delta^3 \cdot (\kappa^2 + \log \Delta))$, where $\Delta$ is logarithmic in the paramters of the problem. A cost query can be answered in constant time and a solution query in time linear in the size of the output.
\end{theorem}

\end{document}